\documentclass[a4paper,reqno]{amsart}

\usepackage{amsmath}
\usepackage{amsfonts}
\usepackage{amssymb}
\usepackage{mathrsfs}
\usepackage{graphicx}
\usepackage{hyperref}
\hypersetup{
  colorlinks   = true, 
  urlcolor     = blue, 
  linkcolor    = blue, 
  citecolor   = blue 
}
\usepackage{enumerate}
\usepackage[inline]{enumitem}

\newtheorem{theorem}{Theorem}
\theoremstyle{plain}

\newtheorem{definition}{Definition}

\newtheorem{lemma}{Lemma}

\newtheorem{proposition}{Proposition}

\numberwithin{equation}{section}

\newcommand{\reals}{\mathbb{R}}
\newcommand{\complexes}{\mathbb{C}}
\newcommand{\integers}{\mathbb{Z}}

\newcommand{\spec}[1]{\mathrm{spec} (#1)}
\newcommand{\iprod}[2]{\langle #1, #2 \rangle}

\newcommand{\tnorm}[1]{\| #1 \|_{1}}
\newcommand{\fprod}[2]{\langle #1, #2 \rangle_{\text{F}}}

\newcommand{\tr}[1]{\mathrm{tr} \, #1 }

\newcommand{\comm}[2]{[#1,#2]}
\newcommand{\acomm}[2]{\{#1,#2\}}
\newcommand{\ad}[1]{\comm{#1}{\,\cdot \,}}

\newcommand{\hadj}{*}
\newcommand{\dual}{\prime}
\newcommand{\id}{\mathrm{id}}

\newcommand{\cp}[1]{\mathrm{CP}( #1 )}
\newcommand{\cptp}[1]{\mathrm{CP}_{\text{t.p.}}( #1 )}
\newcommand{\cpu}[1]{\mathrm{CP}_{\text{u.}}( #1 )}

\newcommand{\matr}[1]{\mathbb{M}_{#1}}
\newcommand{\matrd}{\matr{d}}

\renewcommand{\Re}[1]{\mathrm{Re}\,#1}
\renewcommand{\Im}[1]{\mathrm{Im}\,#1}
\renewcommand{\vec}[1]{\boldsymbol{#1}}

\begin{document}
\title[On the Floquet analysis...]{On the Floquet analysis of commutative periodic Lindbladians in finite dimension}
\author{Krzysztof Szczygielski}
\address
{Institute of Theoretical Physics and Astrophysics, Faculty of Mathematics, Physics and Informatics, University of Gda\'{n}sk, 80-308 Gda\'{n}sk, Poland}%
\email{krzysztof.szczygielski@ug.edu.pl}
\date{\today}

\begin{abstract}
We consider the Markovian Master Equation over matrix algebra $\matrd$, governed by periodic Lindbladian $L_t$ in standard (Kossakowski-Lindblad-Gorini-Sudarshan) form. It is shown that under simplifying assumption of commutativity, i.e.~if $L_t L_{t'} = L_{t'}L_t$ for any moments of time $t,t'\in\reals_+$, the Floquet normal form of resulting completely positive dynamical map is not guaranteed to be given by simultaneously globally Markovian maps. In fact, the periodic part of the solution is even shown to be necessarily non-Markovian. Two examples in algebra $\matr{2}$ are explicitly calculated: a periodically modulated random qubit dynamics, being a generalization of pure decoherence scheme, and a classically perturbed two-level system, coupled to reservoir via standard ladder operators.
\end{abstract}

\maketitle

\section{Introduction}

Periodically controlled open quantum systems recently began gaining increasing attention, mainly for their applicability in quantum information processing, error correction, quantum thermodynamics and general description of dephasing processes in presence of external quasi-classical perturbations. The general microscopic construction of Markovian Master Equation (MME) describing an open quantum system with periodic Hamiltonian, weakly interacting with reservoir of infinite degrees of freedom, was established in \cite{Alicki2006b} and later extended in \cite{Szczygielski2014}. The MME was obtained with application of celebrated \emph{Floquet theory} in the usual regime of \emph{weak coupling limit}. This approach proved itself to be of particular importance for laser spectroscopy and quantum thermodynamics \cite{Szczygielski2013,Alicki2017a,Alicki2018}, ultimately leading to major advancements in description of quantum heat engines, solar cells and related ideas \cite{Alicki2015,Alicki_2015,Alicki_2017,Alicki_2018,Alicki2019}.

In this paper, we elaborate on general properties of the solution of MME on algebra $\matrd$ of complex matrices of size $d$, 
\begin{equation}\label{eq:IntroductionMME}
	\frac{d\rho_t}{dt} = L_t (\rho_t),
\end{equation}
where $\rho_t$ (for $t \in \reals_+ = [0,\infty )$) is a time-dependent \emph{density matrix}, i.e.~a Hermitian, positive semi-definite matrix of trace one, and $L_t$ is the time-periodic and \emph{Lindbladian} in celebrated \emph{standard} (Kossakowski-Lindblad-Gorini-Sudarshan) \emph{form} \cite{Gorini1976,Lindblad1976,Breuer2002,Alicki2006a,Rivas2012},
\begin{equation}\label{eq:Lindbladian}
	L_t (\rho) = -i \comm{H_t}{\rho} + \sum_{j} \left( V_{j,t} \rho V_{j,t}^{\hadj} - \frac{1}{2}\acomm{V_{j,t}^{\hadj} V_{j,t}}{\rho} \right),
\end{equation}
with all matrices $H_t, V_{j,t} \in \matrd$ \emph{periodic} with some period $T > 0$ and $H_t$ being Hermitian (we put $\hbar = 1$ for convenience); $\acomm{a}{b} = ab+ba$ stands for the anticommutator. Lindbladian \eqref{eq:Lindbladian} generates evolution $\rho_t = \Lambda_t (\rho_0)$, where $\{\Lambda_t : t \in \reals_+ \}$ is a one-parameter family of \emph{quantum dynamical maps}, each of them \emph{completely positive}, \emph{trace preserving} and a \emph{trace norm contraction}. Structure of $L_t$ guarantees that $\Lambda_t$ satisfies a much stronger condition of being \emph{CP-divisible}, or \emph{Markovian} \cite{Rivas2012,Chruscinski2014a}:

\begin{definition}\label{def:CPdiv}
Completely positive trace preserving linear map $\Lambda_t$, $t\in\reals_+$ on $\matrd$ will be called \emph{CP-divisible} or \emph{Markovian} over interval $\mathcal{I} \subseteq \reals_+$, if and only if the associated two-parameter family of \emph{propagators}
\begin{equation}
	V_{t,s} = \Lambda_t \Lambda_{s}^{-1}
\end{equation}
is also completely positive and trace preserving for all $t,s \in \mathcal{I}$, $s\leqslant t$.
\end{definition}

In this work, we will be focusing mainly on CP-divisibility (Markovianity) of certain evolution maps understood as in definition \ref{def:CPdiv}. The notation will be fairly standard. From here onwards, we endow $\matrd$ with Frobenius (Hilbert-Schmidt) inner product
\begin{equation}
	\fprod{a}{b} = \tr{a^\hadj b}.
\end{equation}
$\matrd$ is spanned by orthonormal \emph{Frobenius basis} $\{F_i : i = 1, \, ... \,, \, d^2\}$, where we conveniently choose $F_{d^2} = \frac{1}{\sqrt{d}}I$, $I$ standing for identity matrix, with all remaining $F_i$ traceless \cite{Bertlmann2008,Hall2015}. By $\| a \|_1$ and $\| a \|$, $a\in\matrd$, we will respectively denote the \emph{trace norm} and induced \emph{operator norm} of matrix $a$ and its Hermitian conjugate will be $a^\hadj$. For linear space $\mathcal{X}$, $B(\mathcal{X})$ will be the algebra of bounded linear maps over $\mathcal{X}$ (complete with respect to supremum norm). For $A\in B(\mathcal{X})$, symbols $E_{A}(\lambda)$ will denote the eigenspace of $A$ corresponding to eigenvalue $\lambda$ and geometric \emph{multiplicity} of $\lambda$ will be $k_\lambda = \dim{E_{A}(\lambda)} $. To shorten the notation, we will simply write $A \in \cptp{\mathcal{X}}$ (resp. $A \in \cpu{\mathscr{A}}$) if $A$ is \emph{completely positive and trace preserving} over ordered space $\mathcal{X}$ (resp.~\emph{completely positive unital} over unital algebra $\mathscr{A}$). We will write $a \geqslant 0$, or $a \in \mathscr{A}^+$, to indicate that $a$ lays in positive cone in $\mathscr{A}$ (i.e.~is positive semi-definite).

\section{Floquet approach to CP-divisible dynamics}
\label{sec:PeriodicLindbladians}

Since $\matrd$ and $\complexes^{d^2}$ are naturally isomorphic as linear spaces, Master Equation \eqref{eq:IntroductionMME} may be always \emph{vectorized} \cite{Miszczak2011,Am-Shallem2015,IngemarBengtsson2017}, i.e.~represented as ordinary differential equation for vector-valued function $t\mapsto \rho_t \in \complexes^{d^2}$; in such case, $L_t$ is a periodic matrix in $\matr{d^2}$  and hence, the MME is expressible as ordinary differential equation (ODE) with periodic matrix coefficient. Dynamical map $\Lambda_t$ generated by \eqref{eq:IntroductionMME} also admits a bijective representation in $\matr{d^2}$ (often called the \emph{superoperator} in this context) and satisfies a matrix counterpart of MME,
\begin{equation}\label{eq:OperatorMME}
	\frac{d}{dt}\Lambda_t = L_t \Lambda_t , \quad \Lambda_0 = \id,
\end{equation}
being therefore the \emph{principal fundamental solution} of \eqref{eq:IntroductionMME}. Throughout the article, we impose a technical, however also physically justified, additional restriction on regularity of $L_t$: namely, we allow it to change \emph{piecewise-continuously}, but without intermediate jumps at discontinuities, i.e. if some $t_0 \in \reals_+$ is a point of discontinuity of $L_t$, then function $t\mapsto L_t$ will be required to be either left- or right-continuous (in $B(\matrd)$) at $t_0$. By general characterization of ODEs with periodic coefficients provided by celebrated \emph{Floquet's theorem} \cite{Chicone2006}, $\Lambda_t$ admits a product structure
\begin{equation}\label{eq:FloquetFormGeneral}
	\Lambda_t = P_t e^{tX},
\end{equation}
such that function $t\mapsto P_t \in B(\matrd)$ is periodic and absolutely continuous, $P_0 = \mathrm{id}$, and $X\in B(\matrd)$ is constant. Both maps of the pair $(P_t, e^{tX})$ are invertible and the (non-unique) pair itself is called \emph{the Floquet normal form} of solution $\Lambda_t$.

In this section we present some results concerning conditions for CP-divisibility of Floquet normal form, partially in general case (in section \ref{sec:GeneralConsiderations}) and especially, in case of \emph{commutative} Lindbladian families (in section \ref{sec:CommLindbladian}). Commutativity is a severe simplification, however still of practical applicability for various quantum models and of conceptual and mathematical importance, as it provides an exactly solvable case. In particular, we show that one \emph{may not} expect simultaneous CP-divisibility of Floquet pair, even despite their composition is perfectly Markovian quantum dynamics. Some general remarks regarding asymptotic properties of solutions are also addressed (in section \ref{sec:Stability}). Two exemplary applications of such commutative periodic Lindbladian families are then presented in section \ref{sec:Motivation}. 

\subsection{General considerations}
\label{sec:GeneralConsiderations}
Finding the explicit form of a pair $(P_t, e^{tX})$ may be a very challenging task, as it clearly requires one to find an actual solution of the ODE first. In fact, no universal methods of obtaining the solution exist apart from some perturbative approaches, including Dyson, Magnus or Fer expansions \cite{Dyson1949,Blanes1998,Butcher2009}; these however are rarely exactly summable. Some properties may be sometimes deduced from the \emph{stroboscopic} form of the fundamental matrix solution: given a solution $\Lambda_t = P_t e^{tX}$, the stroboscopic dynamics is $\Lambda_{nT} = e^{nTX}$, which is easily implied by periodicity of $P_t$; clearly, $\Lambda_{nT}\in\cptp{\matrd}$. Putting $n=1$, we obtain the so-called \emph{monodromy matrix} $\Lambda_T = e^{TX}$ which allows to find
\begin{equation}\label{eq:XgeneralFormula}
	X = \frac{1}{T} \log{\Lambda_T},
\end{equation}
where existence of the logarithm is assured by invertibility of $\Lambda_T$. The problem arises with complete positivity of a semigroup $\{e^{tX} : t \in \reals_+\}$ as \emph{a priori} there is no guarantee that $\Lambda_T$ lays in the range of any Markovian semigroup, i.e.~any branch of $\log{\Lambda_T}$ in Lindblad form exists. Problem of accessibility of set $\cptp{\matrd}$ (i.e.~quantum channels) by Lindblad semigroups is surprisingly non-trivial even in low-dimensional matrix algebras and is subject to active research \cite{Ruskai2002,Puchala2019,Wolf2008,Schnell2018,Denisov2018}; we will not however address it here directly.

We call a map $\Phi$ on C*-algebra $\mathscr{A}$ a \emph{*-map}, if and only if it is \emph{Hermiticity preserving}, i.e.~satisfies $\Phi(x)^\hadj = \Phi(x^\hadj)$ for all $x\in\mathscr{A}$. The following simple claim holds:

\begin{proposition}\label{prop:SimulHPTP}
Let the Floquet normal form $(P_t, e^{tX})$ satisfy MME \eqref{eq:OperatorMME} for periodic Lindbladian \eqref{eq:Lindbladian}. If one of the maps of the pair is trace preserving and a *-map over $\matrd$, so is the second one.
\end{proposition}

The proof is basic and relies on invertibility of both maps in Floquet pair. Let us now introduce few additional notions. By expanding matrices $V_{j,t}$ in expression \eqref{eq:Lindbladian} in Frobenius basis, one obtains Lindbladian in so-called \emph{first standard form} \cite{Gorini1976,Breuer2002,Alicki2006a,Rivas2012},
\begin{equation}\label{eq:Lindbladian2}
	L_t (\rho) = -i \comm{H_t}{\rho} + \sum_{j,k=1}^{d^2-1} a_{jk}(t) \left( F_j \rho F_{k}^{\hadj} - \frac{1}{2}\acomm{F_{k}^{\hadj} F_j}{\rho}\right),
\end{equation}
where the \emph{Kossakowski matrix} $\mathbf{a}_t = [a_{jk}(t)]\in \matr{d^2 - 1}$ is positive semi-definite and both matrix-valued functions $t\mapsto H_t$, $t\mapsto \mathbf{a}_t$ are piecewise-continuous, as stated earlier. Then, $L_t$ generates a CP-divisible, trace preserving dynamical map iff it is of a form \eqref{eq:Lindbladian}, which is true iff it is of a form \eqref{eq:Lindbladian2} for Hermitian $H_t$ and $\mathbf{a}_t \geqslant 0$. For later use, we also introduce
\begin{equation}\label{eq:DjkDefinition}
	D_{jk}(x) = F_j x F_{k}^{\hadj} - \frac{1}{2}\acomm{F_{k}^{\hadj} F_j}{x} ,
\end{equation}
to shorten the notation a little bit.

Let us assume that periodic part of Floquet normal form $P_t$ is a trace preserving *-map. Applying lemma \ref{lemma:EveryLinMap} (available in \ref{sec:SecondaryResults}), it may be then cast in the form
\begin{equation}\label{eq:PtDecomposition}
	P_t (x) = \sum_{j,k=1}^{d^2} p_{jk}(t) F_j x F_{k}^{\hadj}, \quad x\in\matrd ,
\end{equation}
for Hermitian, periodic matrix $[p_{jk}(t)]\in\matr{d^2}$. This allows us to formulate a following result, which can be considered as a partial answer to the question of simultaneous complete positivity of Floquet pair in general case:

\begin{proposition}\label{prop:PtCPimpliesEtXequiv}
Let $(P_t, e^{tX})$ be the Floquet normal form for $L_t$ \eqref{eq:Lindbladian2} s.t. $P_t$ is a trace preserving *-map over $\matrd$, admitting a form \eqref{eq:PtDecomposition}, and let $\mathbf{\tilde{P}}_t = [p_{jk}(t)]_{j,k=1}^{d^2-1}$. Then $\{e^{tX} : t\in\reals_+\}$, is a Markovian contraction semigroup if and only if
\begin{equation}\label{eq:XstandardCondition}
	\mathbf{a}_0 - \left.\frac{d\mathbf{\tilde{P}}_{t}}{dt}\right|_{0} \in \matr{d^2-1}^{+}.
\end{equation}
\end{proposition}

\begin{proof}
Showing the claim involves simple algebra, therefore we only sketch the proof. As $\Lambda_t$ satisfies the operator MME \eqref{eq:OperatorMME}, after differentiating $\Lambda_t = P_t e^{tX}$ one easily obtains $\frac{dP_t}{dt} = L_t P_t - P_t X$ which, after putting $t=0$ and reordering, yields
\begin{equation}\label{eq:XformulaL0P0}
	X = L_0 - \left.\frac{dP_t}{dt}\right|_{0}.
\end{equation}
Therefore, $\{e^{tX} : t \in \reals_+\}$ is CP-divisible if and only if \eqref{eq:XformulaL0P0} is of standard form. According to lemma \ref{lemma:EveryLinMap}, $P_t$ can be given as
\begin{equation}\label{eq:Pt2}
	P_t (x) = x + i \comm{G_t}{x} - \acomm{K_t}{x} + \sum_{j,k=1}^{d^2 - 1} p_{jk}(t) F_j x F_{k}^{\hadj},
\end{equation}
where $G_t$ and $K_t$ are Hermitian matrices,
\begin{equation}\label{eq:Gt}
	G_t = \frac{1}{2i\sqrt{d}} \sum_{j=1}^{d^2-1} \Big( p_{jd^2} (t)F_j - p_{d^2 j}(t) F_{j}^{\hadj} \Big), \qquad K_t = \frac{1}{2}\sum_{j,k=1}^{d^2-1} p_{jk}(t) F_{k}^{\hadj} F_j .
\end{equation}
Differentiating \eqref{eq:Pt2} and substituting to \eqref{eq:XformulaL0P0} leads, after some algebra, to
\begin{equation}\label{eq:Xexplicit}
	X(x) = -i\comm{H_0 - \left.\frac{dG(t)}{dt}\right|_{0}}{x} + \sum_{j,k=1}^{d^2-1} \left( a_{jk}(0) - \left.\frac{dp_{jk}(t)}{dt}\right|_{0} \right) D_{jk}(x),
\end{equation}
for $D_{jk}$ given via \eqref{eq:DjkDefinition}. Matrix $\mathbf{\tilde{P}}_t$ is clearly Hermitian, so $H_0 - \left.\frac{dG(t)}{dt}\right|_{0}$ is also; hence, \eqref{eq:Xexplicit} defines a generator of completely positive contraction semigroup, i.e.~is Markovian, if and only if a matrix $\left[ a_{jk}(0) - \left.\frac{dp_{jk}(t)}{dt}\right|_{0} \right]_{jk} = \mathbf{a}_0 - \left. \frac{d\mathbf{\tilde{P}}_t}{dt}\right|_{0}$ is positive semi-definite.
\end{proof}

\subsection{Stability and asymptotic behavior of solutions}
\label{sec:Stability}

Stability of solutions remains a significant matter of classical theory of ODEs. Naturally, it is equally important in context of Floquet analysis as we are very often interested in qualitative \emph{asymptotic} behavior of solutions to certain initial value problems, i.e.~after very long evolution time. In particular, asymptotic behavior of Floquet solutions is fully deducible from analysis of so-called \emph{characteristic multipliers} of the system (see below) and it is known that solutions exhibit dramatically different characteristics depending on the multipliers, ranging from almost-exponential decaying to 0, through formation of periodic limit cycles to even unbounded growth, or ``blowing up'', as $t\to\infty$. Fortunately, in our case of Markovian dynamics, the infinite growth scenario is forbidden (loosely speaking, by contractivity of dynamical maps), however other possibilities remain.

Let us now assume that $X$, given in the Floquet normal form is \emph{diagonalizable}, i.e.~satisfies an eigenequation $X(\varphi_j) = \mu_j \varphi_j$ for $\mu_j \in \complexes$, $j \in \{1,\,2,\, ... \, , \, d^2\}$, such that set $\{\varphi_j\}$ is linearly independent and hence a \emph{basis} in $\matrd$. Monodromy matrix $\Lambda_T = e^{TX}$ satisfies the eigenequation for the same set of matrices,
\begin{equation}
	\Lambda_T (\varphi_j) = \lambda_j \varphi_j, \quad \lambda_j = e^{\mu_j T}, \quad j \in \{1,\,2,\, ... \, , \, d^2\},
\end{equation}
and by spectral mapping theorem, $e^{tX}(\varphi_j) = e^{\mu_j t}\varphi_j$, $t\in\reals_+$. We then call $\spec{\Lambda_T} = \{\lambda_j\}$ the set of \emph{characteristic multipliers} and $\spec{X} = \{\mu_j\}$ the set of \emph{characteristic exponents} of the system. Note that $\{\mu_j\}$ is not uniquely defined by monodromy matrix, as shifting transformation $\{\mu_j\} + \frac{2\pi i}{T} \vec{k}$, $\vec{k}\in\integers^{d^2}$, leaves $\spec{\Lambda_T}$ unchanged (simply, $\log{\Lambda_T}$ is non-unique). Now, define a set of functions
\begin{equation}
	\rho_j (t) = \Lambda_t (\varphi_j) =e^{\mu_j t}\phi_j (t), \quad \phi_j (t) = P_t (\varphi_j),
\end{equation}
which are naturally \emph{solutions} to the MME in question, i.e.~\emph{states}. By diagonalizability of $X$, set $\{\rho_j\}$ is a \emph{fundamental} set of solutions. The general solution for \eqref{eq:IntroductionMME} is then expressible as a linear combination
\begin{equation}\label{eq:FloquetGenSol}
	\rho_t = \sum_{j=1}^{d^2} c_j \rho_j (t) = \sum_{j=1}^{d^2} c_j e^{\mu_j t} \phi_j (t),
\end{equation}
where coefficients $c_j$ are prescribed by initial condition $\rho_0 = \sum_{j=1}^{d^2} c_j \varphi_j$. Otherwise, if $X$ is considered non-diagonalizable, a fundamental set of solutions loses the above simple structure and reflects the Jordan normal form of $X$; see e.g. \cite{Yakubovich1975} for further details. Evidently, by periodicity of \emph{Floquet states} $\phi_j (t)$, we have also $\rho_j (t+nT) = \lambda_{j}^{n} \rho_j (t)$. As a result, the stroboscopic dynamics $\Lambda_{nT}$ simply \emph{multiplies} the initial state $\rho_j (0)$ by factor $\lambda_{j}^{n} = e^{n\mu_j T}$ and a long-time behavior of solution is directly influenced by properties of characteristic multipliers ($\mathbb{S}^{1}$ denotes a unit circle in $\complexes$):

\begin{proposition}\label{prop:Stability}
The long-time behavior of solution $\rho_j (t)$ is determined by the characteristic multiplier $\lambda_j$ in a following way \cite{Yakubovich1975}:
\begin{itemize}
	\item If $|\lambda_j|<1$, then $\rho_j (t)$ vanishes as $t\to\infty$.
	\item If $\lambda_j = 1$, then $\rho_j (t)$ is \emph{periodic}. If, on the other hand $\lambda_j \in \mathbb{S}^{1} \setminus \{1\}$, then $\rho_j (t)$ is \emph{pseudo-periodic}, i.e. satisfies equality $\rho_j (t+T) = e^{i\theta}\rho_j (t)$ for some $\theta \in [0,2\pi )$; in particular, for $\lambda_j = -1$, solution $\rho_j (t)$ flips a sign, $\rho_j (t+T) = -\rho_j (t)$, which is sometimes referred to as \emph{anti-periodicity}.
	\item If $|\lambda_j|>1$, then $\rho_j (t)$ grows infinitely in norm.
\end{itemize}
\end{proposition}

Naturally, $|\lambda_j|\leqslant 1$ guarantees that the solution is \emph{stable} and if $|\lambda_j|>1$, \emph{unstable} (``blows up'' at large times); hence, a general solution \eqref{eq:FloquetGenSol} will be called \emph{asymptotically stable}, if and only if all $|\lambda_j| \leqslant 1$. Fortunately, in case of quantum dynamics, unstability of solutions is disallowed by spectral properties of monodromy matrix:

\begin{proposition}\label{prop:AllStable}
Spectrum of $\Lambda_T$ lays inside unit disc $\mathbb{D}^1$ and is invariant with respect to complex conjugation. In the result, all solutions of MME \eqref{eq:IntroductionMME} for periodic $L_t$ in standard form \eqref{eq:Lindbladian2} are asymptotically stable.
\end{proposition}

\begin{proof}

Stability is a straightforward consequence of a known fact, that the spectral radius of completely positive and trace preserving map is exactly 1. This easily follows from spectral properties of a dual map $\Lambda_{T}^{\dual}$ which is necessarily unital and completely positive on $\matrd$ and attains its norm at $I$ \cite[Proposition 3.6]{Paulsen2003}. As $\Lambda_{T}^{\dual}$ is also a \emph{*-map}, taking the Hermitian adjoint of eigenequation $\Lambda_{T}^{\dual}(x) = \lambda x$ for any $\lambda\in\spec{\Lambda_{T}^{\dual}}\setminus\{1\}$ and some $x\in\matrd$, yields $\overline{\lambda}$ is also an eigenvalue for eigenvector $x^{\hadj}$. This shows that $\spec{\Lambda_{T}^{\dual}}\setminus \{1\}$ is either real or consists of pairs $\{\lambda,\overline{\lambda}\}$, $|\lambda|\leqslant 1$, i.e.~$\spec{\Lambda_{T}}$ is invariant w.r.t. complex conjugation. This finally shows that solutions of a form $\rho_j (t) = e^{\mu_j t}\phi_j (t)$, as well as any general solution \eqref{eq:FloquetGenSol}, are all stable by proposition \ref{prop:Stability}. 
\end{proof}

\begin{proposition}\label{prop:VarphiProperties}
The following claims hold:
\begin{enumerate}
	\item\label{claim:VarphiPropertiesOne} If $\lambda\neq 1$, then $\tr{\varphi} = 0$ and $\varphi$ is not positive semi-definite;
	\item\label{claim:VarphiPropertiesThree} There exists $\varphi\in E_{\Lambda_T}(1)$, such that $\varphi\geqslant 0$;
	\item\label{claim:VarphiPropertiesFour} If $\varphi \in E_{\Lambda_T}(\lambda)$ for $\lambda\in\spec{\Lambda_T}\setminus\reals$, then $\varphi^{\hadj}\in E_{\Lambda_T}(\overline{\lambda})$. If eigenvalue $\lambda\in\reals$ is simple (i.e.~$k_\lambda = 1$) then $\varphi$ is Hermitian.
\end{enumerate}
\end{proposition}

\begin{proof}
For claim \ref{claim:VarphiPropertiesOne}, note that as $\Lambda_T$ satisfies eigenequation $\Lambda_T (\varphi) = \lambda\varphi$, trace preservation condition demands
\begin{equation}\label{eq:TPcondition}
	(1-\lambda)\, \tr{\varphi} = 0.
\end{equation}
Let $\lambda\neq 1$ and thus $\tr{\varphi} = 0$. Assume indirectly $\varphi \geqslant 0$; then its trace norm $\tnorm{\varphi} = \tr{\varphi} = 0$, which is possible if and only if $\varphi = 0$, a contradiction; therefore any eigenvector $\varphi$ corresponding to eigenvalue $\lambda\neq 1$ must not be positive semi-definite. For claim \ref{claim:VarphiPropertiesThree}, we will utilize another known result which states that \emph{if $\Phi$ is positive on finite dimensional C*-algebra $\mathscr{A}$ and $r$ is its spectral radius, then there exists eigenvector $x\in\mathscr{A}$ such that $\Phi(x) = r x$ and $x \geqslant 0$} \cite[Theorem 2.5]{Evans1978}. As spectral radius of $\Lambda_{T}$ is 1, demanding $\lambda = 1$ yields that eigenequation $\Lambda_T (\varphi)=\varphi$ is satisfied for (at least one) matrix $\varphi\geqslant 0$. Finally, claim \ref{claim:VarphiPropertiesFour} is a direct consequence of Hermiticity preservation: given $\varphi\in E_{\Lambda_T}(\lambda)$, the adjoint of eigenequation $\Lambda_T (\varphi) = \lambda \varphi$ gives $\varphi^\hadj \in E_{\Lambda_T}(\overline{\lambda})$. If $\lambda\in\reals$, then $\varphi, \varphi^\hadj \in E_{\Lambda_T}(\lambda)$ and $k_\lambda \geqslant 2$. If $\lambda$ is simple, then $\dim{E_{\Lambda_T}(\lambda)}=1$ and eigenvectors $\varphi$, $\varphi^\hadj$ must be linearly dependent, which is possible only if they are equal.
\end{proof}

Finally, we summarize by noticing that in certain scenario, completely positive dynamics will always admit a periodic steady state:

\begin{theorem}\label{thm:PeriodicSteadyState}
Each solution $\rho_t$ becomes arbitrarily close to a certain function $t\mapsto \rho_{t}^{\infty}$, uniformly in space $\mathcal{C}_0([t_0,\infty),\matrd)$ of continuous matrix-valued functions, for $t_0 \geqslant 0$ large enough. If in addition
\begin{equation}\label{eq:PeriodicSteadyStateCond}
	\spec{\Lambda_T}\cap (\mathbb{S}^1 \setminus \{1\}) = \emptyset ,
\end{equation}
then $\rho_{t}^{\infty}$ is an asymptotic periodic limit cycle, i.e.~a periodic steady state.
\end{theorem}

\begin{proof}
Properties of spectrum of monodromy matrix allow to decompose $\spec{X}$ into four disjoint subsets, $\spec{X} = \mathcal{E}_{1,\text{e}}\cup\mathcal{E}_{1,\text{o}}\cup\mathcal{E}_{2}\cup\mathcal{E}_{3}$, such that
\begin{align}\label{eq:SpectralDecLambdaT}
	&\mathcal{E}_{1,\text{e}} = \{\mu_j = \frac{2k_j \pi i}{T}, \, k_j \in\integers\}, \quad \mathcal{E}_{1,\text{o}} = \{\mu_j = \frac{(2k_j+1)\pi i}{T}, \, k_j \in\integers\}, \\ &\mathcal{E}_2 = i\reals \setminus (\mathcal{E}_{1,\text{e}}\cup\mathcal{E}_{1,\text{o}}), \quad \mathcal{E}_3 = \{\Re{\mu_j} < 0\}.\nonumber
\end{align}
Of course the function $\rho^{\infty}_t$ is constructed by re-grouping terms in expression \eqref{eq:FloquetGenSol} and deprecating the sum over set $\mathcal{E}_3$,
\begin{align}\label{eq:RhoPeriodicSteady}
	\rho^{\infty}_{t} &= \sum_{\mu_j \in \mathcal{E}_{1,\text{e}}} c_j e^{2i k_j \pi t / T} \phi_j (t) + \sum_{\mu_j \in \mathcal{E}_{1,\text{o}}} c_j e^{i(2k_j+1) \pi t / T} \phi_j (t) \\
	&+ \sum_{\mu_j \in \mathcal{E}_2} c_j e^{i \, \Im{\mu_j} t} \phi_j (t).\nonumber
\end{align}
Indeed, direct calculations allow to estimate
\begin{equation}
	\tnorm{\rho_t - \rho_{t}^{\infty}} = \left\| \sum_{\mu_j\in\mathcal{E}_3}c_j e^{\mu_j t} \phi_{j}(t) \right\|_{1} \leqslant A e^{-at},
\end{equation}
where $a = \max{|\Re{\mu_j}|}$ and $A$ is a positive constant. Taking any $\epsilon > 0$, one checks that for $t_0 = \frac{1}{a} \ln{\frac{A}{\epsilon}}$ we have $\sup_{t\geqslant t_0}{\tnorm{\rho_t - \rho_{t}^{\infty}}} \leqslant \epsilon$, i.e.~functions $\rho_t$, $\rho_{t}^{\infty}$ are indeed arbitrarily close to each other in uniform topology in $\mathcal{C}_0([t_0,\infty),\matrd)$.

The first sum in \eqref{eq:RhoPeriodicSteady} is \emph{periodic} and the second one is \emph{anti-periodic} (flips a sign after every time shift by $T$); every term appearing in third sum is \emph{pseudo-periodic} (as time-shifting by $T$ shifts coefficients $c_j$ by phase factors, $c_j \mapsto c_j e^{i\, \Im{\mu_j}T}$). Now, if condition \eqref{eq:PeriodicSteadyStateCond} is satisfied then one automatically has $\mathcal{E}_{1,\text{o}}=\mathcal{E}_{2}=\emptyset$ and only the periodic part of \eqref{eq:RhoPeriodicSteady} remains.
\end{proof}

\section{Commutative Lindbladian families}
\label{sec:CommLindbladian}

Here we inspect a simplified class of \emph{commutative} Lindbladian, which provides an exactly solvable case. We assume that the family $\{L_t : t \in \reals_+\}$ of periodic Lindbladians in standard form \eqref{eq:Lindbladian2} satisfies commutativity condition
\begin{equation}\label{eq:CommCondition}
	L_t L_s (x) = L_s L_t (x), \quad t,s \in \reals_+, \, x \in \matrd .
\end{equation}
\subsection{CP-divisibility of Floquet normal form}

The core result of this section, presented in form of theorems \ref{prop:etXcpDiv} and \ref{prop:NotEverywhereCPdiv} below, shows that for special case of commutative Lindbladians \eqref{eq:CommCondition}, both maps of Floquet pair $(P_t, e^{tX})$ can be simultaneously Markovian over some intervals in $\reals_+$ and the semigroup part $e^{tX}$ in fact is Markovian in whole $\reals_+$. However, it is not true for the periodic part $P_t$ as an interesting property is revealed: it is \emph{impossible} for $P_t$ to be uniformly Markovian over a whole time of evolution. The question of simultaneous CP-divisibility of Floquet pair, stated in the Introduction, is hence answered negatively.

\begin{theorem}\label{prop:etXcpDiv}
Let $L_t$ be of standard form \eqref{eq:Lindbladian2}, periodic and obeying the commutativity condition \eqref{eq:CommCondition}. Then, it generates a CP-divisible quantum dynamical map $\Lambda_t$ admitting Floquet normal form $(P_t , e^{tX})$ such that:
\begin{enumerate}
	\item\label{claim2} $\{e^{tX} : t\in\reals_+ \}\subset \cptp{\matrd}$ and is CP-divisible contraction semigroup (i.e.~a quantum dynamical semigroup);
	\item\label{claim1} $P_t$, $t\in\reals_+$, is a trace preserving *-map;
	\item\label{claim:PtCPdivcondition} $P_t$ is CP-divisible in interval $\mathcal{I}\subset\reals_+$ if and only if
	\begin{equation}\label{eq:PtCPdivcondition}
		\mathbf{a}_t - \frac{1}{T} \int_{0}^{T} \mathbf{a}_{t^\prime} dt^\prime \in \matr{d^2-1}^{+} \quad \text{for every }  t\in \mathcal{I};
	\end{equation}
	\item\label{claim:PtCPcondition} $P_t$ is completely positive for some $t\in\reals_+$, if
	\begin{equation}\label{eq:PtCPsuffCond}
		\int_{0}^{t} \mathbf{a}_{t^\prime} dt^\prime - \frac{t}{T} \int_{0}^{T} \mathbf{a}_{t^\prime} dt^\prime \in \matr{d^2-1}^{+}.
	\end{equation}
\end{enumerate}
\end{theorem}

\begin{theorem}\label{prop:NotEverywhereCPdiv}
Map $P_t$ governed by Lindbladian \eqref{eq:Lindbladian2} satisfies the following:
\begin{enumerate}
	\item\label{claim:NotEverywhereCPdivOne} $P_t$ is CP-divisible everywhere in $\reals_+$ iff Kossakowski matrix $\mathbf{a}_t$ is constant;
	\item\label{claim:NotEverywhereCPdivTwo} If $\mathbf{a}_t$ is constant, then $P_t \in\cptp{\matrd}$ for all $t\in\reals_+$;
	\item\label{claim:NotEverywhereCPdivThree} If $\mathbf{a}_t$ is non-constant, then there exists a non-empty union of intervals $\mathcal{N} \subset \reals_+$ such that $P_t$ is not CP-divisible (non-Markovian) in $\mathcal{N}$.
\end{enumerate}
\end{theorem}

\begin{proof}[Proof of theorem \ref{prop:etXcpDiv}]
Commutativity condition \eqref{eq:CommCondition} allows to avoid cumbersome \emph{time-ordering} procedure (like in Dyson expansion) and solution to MME \eqref{eq:IntroductionMME} is exactly obtainable. For brevity, let us introduce three antiderivatives
\begin{equation}\label{eq:AntiderDef}
	\mathcal{H}_t = \int_{0}^{t}H_{t^\prime} dt^{\prime}, \quad A_{jk}(t) = \int_{0}^{t} a_{jk}(t^\prime) dt^\prime, \quad \mathbf{A}_t = [A_{jk}(t)]_{jk} = \int_{0}^{t} \mathbf{a}_{t^\prime} dt^{\prime}.
\end{equation}
Define map $\Phi_t$ on $\matrd$ via
\begin{equation}
	\Phi_t = \exp{\int_{0}^{t} L_{t^\prime} dt^{\prime}} = \exp{\Bigg(-i\ad{\mathcal{H}_t} + \sum_{j,k=1}^{d^2-1} A_{jk}(t) D_{jk} \Bigg)}.
\end{equation}
Then, by direct calculation one can check, by expanding matrix exponentials into power series and applying commutativity condition \eqref{eq:CommCondition}, that $L_t$ commutes with $\Phi_t$ and $\Phi_t$ satisfies differential equation
\begin{equation}
	\frac{d}{dt}\Phi_t = \Phi_t L_t = L_t \Phi_t, \quad \Phi_0 = \id,
\end{equation}
which is simply the MME in question; hence we have $\Phi_t = \Lambda_t$ as $\Phi_t$ must be a unique solution and the monodromy matrix is $\Lambda_T = \exp{\int_{0}^{T}L_{t}dt}$. Finding $X$ requires one to solve an equation $\Lambda_T = e^{TX}$ by computing a logarithm of monodromy matrix (which is achieved by seeking for Jordan normal form of $\Lambda_T$; see e.g.~\cite{Higham2008} for details), which cannot be uniquely determined. In effect, one obtains an infinite family of valid logarithms; for our purpose however, it totally suffices to choose
\begin{equation}\label{eq:Xformula}
	X = \frac{1}{T} \int_{0}^{T} L_{t} dt = -\frac{i}{T} \ad{\mathcal{H}_{T}} + \frac{1}{T} \sum_{j,k=1}^{d^2-1} A_{jk}(T) D_{jk}.
\end{equation}
Clearly, $\mathcal{H}_{T}$ is Hermitian. Moreover, for any $\vec{x} = (x_i) \in \complexes^{d^2 - 1}$ and $t\in\reals_+$,
\begin{equation}
	\iprod{\vec{x}}{\mathbf{A}_t \vec{x}} = \sum_{j,k=1}^{d^2-1} A_{jk}(t) x_j \overline{x_k} = \int_{0}^{t} \Bigg( \sum_{j,k=1}^{d^2 - 1} a_{jk}(t^\prime) x_j \overline{x_k} \Bigg) dt^{\prime} \geqslant 0,
\end{equation}
since $[a_{jk}(t)]_{jk} \geqslant 0$; therefore, also $\mathbf{A}_t \geqslant 0$ for all $t\in\reals_+$ and $X$ chosen in \eqref{eq:Xformula} is of \emph{standard form}. In other words, if commutativity condition holds then there always exists map $X$ solving equation $\Lambda_T = e^{TX}$, which generates a Markovian semigroup; this proves claim \ref{claim2}. For claim \ref{claim1}, note that since $\{e^{tX} : t\in\reals_+ \}$ is a Markovian dynamics, then $P_t$ is also trace preserving *-map via proposition \ref{prop:SimulHPTP}.

By formula \eqref{eq:Xformula}, $X$ commutes with any integral of a form $\int_{t_1}^{t_2} L_{t^\prime} dt^{\prime}$ and therefore $\Lambda_t X = X \Lambda_t$. This in turn implies
\begin{equation}\label{eq:CommCasePtgeneral}
	P_t = \Lambda_t e^{-tX} = \exp{\int_{0}^{t} \left(L_{t^\prime}-X \right) dt^\prime}
\end{equation}
which further yields an explicit formula for $P_t$,
\begin{equation}\label{eq:CommCasePtgeneral2}
	P_t = \exp{\left[ -i \ad{\mathcal{H}_t - \frac{t}{T}\mathcal{H}_T} + \sum_{j,k=1}^{d^2-1} \left( A_{jk}(t)-\frac{A_{jk}(T)}{T}t\right) D_{jk} \right]}.
\end{equation}
By inspection, $P_t$ is clearly periodic. To show claim \ref{claim:PtCPdivcondition}, note that \eqref{eq:CommCasePtgeneral} implies
\begin{equation}
	\frac{dP_t}{dt} = (L_t - X)P_t, \quad P_0 = \id ,
\end{equation}
since $X$ and $P_t$ commute. By general considerations \cite{Rivas2012,Chruscinski2014a}, if some map $\Phi_t$ satisfies an ODE of a form $\frac{d}{dt}\Phi_t = \mathcal{G}_t \Phi_t$, then $\Phi_t$ is CP-divisible in interval $\mathcal{I}\subseteq\reals_+$ if and only if $\mathcal{G}_t$ is of standard form for every $t\in\mathcal{I}$. This shows that sufficient and necessary condition for CP-divisibility of $P_t$ is
\begin{equation}\label{eq:LtminusX}
	L_t - X = -i \ad{H_t - \frac{1}{T}\mathcal{H}_T} + \sum_{j,k=1}^{d^2-1} \left( a_{jk}(t)-\frac{A_{jk}(T)}{T}\right) D_{jk}
\end{equation}
being of standard form which, by obvious hermiticity of $H_t - \frac{1}{T}\mathcal{H}_T$, leads to condition \eqref{eq:PtCPdivcondition}. Finally, claim \ref{claim:PtCPcondition} is a direct consequence of the fact that under condition \eqref{eq:PtCPsuffCond} the map $P_t$ given by \eqref{eq:CommCasePtgeneral2} is an exponential of a standard form Lindbladian for given $t\in\reals_+$ and as such, must be completely positive. We note, that alternatively one can prove this fact directly by appropriately putting $P_t$ in Choi-Kraus form in a fashion similar to the proof of \cite[Theorem 4.2.1]{Rivas2012}; we omit this computation here, however.
\end{proof}

\begin{proof}[Proof of theorem \ref{prop:NotEverywhereCPdiv}]
Notice, that if $\mathbf{a}_t \geqslant 0$ is constant, conditions \eqref{eq:PtCPdivcondition} and \eqref{eq:PtCPsuffCond} given in theorem \ref{prop:etXcpDiv} are automatically satisfied so $P_t \in \cptp{\matrd}$ and is CP-divisible everywhere; this proves claim \ref{claim:NotEverywhereCPdivTwo} as well as necessity stated in claim \ref{claim:NotEverywhereCPdivOne}. For sufficiency, let us assume $P_t$ is CP-divisible everywhere in $[0,T)$ (and in $\reals_+$ in consequence). Then, for any $\vec{x} \in \complexes^{d^2 -1}$, define non-negative piecewise continuous function $f_{\vec{x}}(t)$ by
\begin{equation}
	f_{\vec{x}}(t) = \iprod{\vec{x}}{\mathbf{a}_t \vec{x}} = \sum_{j,k=1}^{d^2-1} a_{jk}(t) x_j \overline{x_k}
\end{equation}
and denote its restriction to $[0,T)$ by the same symbol. Everywhere CP-divisibility of $P_t$ yields, by theorem \ref{prop:etXcpDiv}, that condition \eqref{eq:PtCPdivcondition} is met for every $t\in [0,T)$, i.e.
\begin{equation}\label{eq:CPdivassump}
	f_{\vec{x}}(t) - \frac{1}{T} \int_{0}^{T} f_{\vec{x}}(t) dt \geqslant 0 \quad \text{for all }\vec{x}\in \complexes^{d^2-1}.
\end{equation}
Take any $\vec{x} \neq 0$. By the mean value theorem for definite integrals we have $\frac{1}{T}\int_{0}^{T} f_{\vec{x}}(t) dt = A$ for some $A\geqslant 0$ which satisfies
\begin{equation}\label{eq:CPdivassump2}
	\inf_{t\in [0,T)}{f_{\vec{x}}(t)} \leqslant A \leqslant \sup_{t\in [0,T)}{f_{\vec{x}}(t)}.
\end{equation}
Therefore, by introducing function $g_{\vec{x}}(t) = f_{\vec{x}}(t) - A$, condition \eqref{eq:CPdivassump} may be simply rewritten as $g_{\vec{x}}(t) \geqslant 0$ for all $t\in\reals_+$. This however implies $A$ is a lower bound for $f_{\vec{x}}$ and, by \eqref{eq:CPdivassump2}, $A = \inf_{t\in [0,T)}{f_{\vec{x}}(t)}$. This implies
\begin{equation}\label{eq:fxAcondition}
	\frac{1}{T} \int_{0}^{T} g_{\vec{x}}(t) \, dt = \frac{1}{T} \int_{0}^{T} \left[f_{\vec{x}}(t) - A \right] dt = 0.
\end{equation}
By the initial assumptions on regularity of $L_t$, function $f_{\vec{x}}$ is piecewise-continuous and either left- or right-continuous at every discontinuity point. The set of all discontinuity points provides a partition $(\Delta_j)$ of $[0,T)$ of mutually disjoint intervals, $\Delta_j \cap\Delta_{j+1}=\emptyset$, $\bigcup_j \Delta_j = [0,T)$ which are either \emph{open}, \emph{closed} or \emph{half-open} such that every discontinuity point $t_0$ belongs either to $\Delta_j$, or $\Delta_{j+1}$. Then, piecewise-continuity of $f_{\vec{x}}$ allows it to be represented as
\begin{equation}
	f_{\vec{x}} (t) = \sum_{j} \xi_{\vec{x}}^{(j)}(t) \chi_{j} (t)
\end{equation}
where functions $\xi_{\vec{x}}^{(j)}$ are continuous and $\chi_{j}$ stands for the indicator function of interval $\Delta_j$, i.e.~$\chi_j (t) = 1$ iff $t\in\Delta_j$ and 0 otherwise. Then, \eqref{eq:fxAcondition} implies
\begin{equation}
	\sum_j \int_{\Delta_j} \left[ \xi_{\vec{x}}^{(j)}(t) \, dt - A \right] dt = 0,
\end{equation}
which is possible iff all $\int_{\Delta_j} \left[ \xi_{\vec{x}}^{(j)}(t) \, dt - A \right] dt = 0$. Since every function $\xi_{\vec{x}}^{(j)}$ is continuous everywhere inside $\Delta_j$, we have $\xi_{\vec{x}}^{(j)}(t) = A$ for all $t\in\operatorname{Int}{\Delta_j}$. For any discontinuity point $t_0 \in\reals_+$, assume (with no loss of generality) that $t_0$ is a right boundary of some right-closed interval $\Delta_j$; since $\xi_{\vec{x}}^{(j)}$ is assumed to be left-continuous at $t_0$, it must also be that $\xi_{\vec{x}}^{(j)}(t_0)=A$ (analogous reasoning then is true for right-continuous case) and so $f_{\vec{x}}(t) = A$ everywhere, i.e.~$\mathbf{a}_t$ is constant and claim \ref{claim:NotEverywhereCPdivOne} is shown. Finally, for claim \ref{claim:NotEverywhereCPdivThree}, assume $\mathbf{a}_t$ is not constant. Then, $P_t$ is not everywhere CP-divisible via claim \ref{claim:NotEverywhereCPdivOne}, or equivalently, inequality \eqref{eq:CPdivassump} is not satisfied for all $\vec{x}\in\complexes^{d^2 -1}$. Denote now
\begin{equation}
	\mathcal{P}_{\vec{x}} = \{ t\in [0,T) : f_{\vec{x}}(t) \geqslant \frac{1}{T} \int_{0}^{T} f_{\vec{x}}(t) dt \}, \qquad \mathcal{P} = \bigcap_{\vec{x}\in\complexes^{d^2-1}}\mathcal{P}_{\vec{x}}.
\end{equation}
Under such notion, CP-divisibility of $P_t$ is allowed only over subset $\mathcal{P}\subsetneq [0,T)$ and hence, its complement $\mathcal{N} = [0,T) \setminus \mathcal{P}$ is non-empty. By piecewise continuity of $f_{\vec{x}}$, both $\mathcal{P}$ and $\mathcal{N}$ must be unions of intervals in $[0,T)$.
\end{proof}

\section{Exemplary applications}
\label{sec:Motivation}

In this section, we examine two examples of Master Equations governed by \emph{commutative periodic Lindbladian families}. For clarity of presentation, we will limit our analysis to the simplest case of algebra $\matr{2}$, however generalizations to higher dimensional systems are naturally obtainable. In all the following, the orthonormal Frobenius basis in $\matr{2}$ is then $F_j = \frac{\sigma_j}{\sqrt{2}}$, where $\{\sigma_{j}\}_{j=1}^{4}$ are the Pauli matrices. The solutions of differential equations over $\matr{2}$ appearing in this section will always be obtained by the so-called \emph{vectorization} procedure, i.e.~by applying some arbitrarily chosen isomorphism $\matr{2} \mapsto \complexes^{4}$. For simplicity, we choose it as
\begin{equation}
	x \mapsto \vec{x} = (x_1, \, x_2 ,\, x_3 ,\, x_4)^{\mathrm{T}}, \quad x_j = \frac{1}{\sqrt{2}}\,\tr{\sigma_j x},
\end{equation}
i.e.~we map each matrix to a vector of its components in Frobenius basis. Note, that $x_4 = \frac{1}{\sqrt{d}}\tr{x}$. Then, every map $W \in B(\matr{2})$ is then expressed as a matrix
\begin{equation}
	\mathbf{W} = [W_{jk}]_{jk} \in \matr{4}, \quad W_{jk} = \frac{1}{2} \, \tr{\left[\sigma_j W(\sigma_k)\right]}.
\end{equation}
In particular, $W$ is trace preserving iff $W_{d^2 j} = \delta_{d^2 j}$. If a Hermitian basis $\{F_i\}$ is used (which is the case here), then $W$ is a *-map iff $[W_{jk}]_{jk}$ is real. Likewise, we make bijective replacements $\rho_t \mapsto \vec{r}(t)$, $\Lambda_t \mapsto \mathbf{\Lambda}(t) = [\Lambda_{jk}(t)]_{jk}$ and $L_t \mapsto \mathbf{L}(t) = [L_{jk}(t)]_{jk}$, such that the MME transforms into linear ODE of a form
\begin{equation}
\label{eq:VecODE}
	\frac{d\vec{r}(t)}{dt} = \mathbf{L}(t) \vec{r}(t).
\end{equation}

\subsection{Periodically modulated random dynamics}
\label{sec:PeriodicRandomDyn}

As a first simple, yet popular example, we will briefly analyze a random dynamics with additional assumption of time-periodicity of decoherence rates, i.e.~a generalization of pure decoherence model of a qubit, involving all Pauli channels. We take the Master Equation in a following form \cite{Chruscinski2013}
\begin{equation}\label{eq:pureDecohMME}
	\frac{d\rho_t}{dt} = L_t (\rho_t) = \frac{1}{2}\sum_{j=1}^{3} \gamma_j (t) (\sigma_j \rho_t \sigma_j - \rho_t) .
\end{equation}
We assume all functions $\gamma_j (t)$ are \emph{non-negative}, \emph{periodic} and \emph{continuous}. Exploiting a useful property of Pauli matrices $\sigma_{j}^{2} = \sigma_{j}^{\hadj}\sigma_j = I$, \eqref{eq:pureDecohMME} is quickly seen to be of form \eqref{eq:Lindbladian2} for $H_t = 0$ and Kossakowski matrix $\mathbf{a}_t = [\delta_{jk}\gamma_{j}(t)]_{jk}$. In such case, the derived $\Lambda_t$ is a convex combination of Pauli channels. 

Invoking the vectorization procedure mentioned earlier, matrix $\mathbf{L}(t)$ is found to be diagonal in Frobenius basis,
\begin{equation}
	\mathbf{L}(t) = - \, \mathrm{diag}\{\gamma_{2}(t)+\gamma_{3}(t), \, \gamma_{1}(t)+\gamma_{3}(t), \, \gamma_{1}(t)+\gamma_{2}(t), \, 0\}.
\end{equation}
Note, that $L_{44}(t)=0$ which is required for trace preservation. Solution to \eqref{eq:pureDecohMME} is then again given by diagonal matrix $\mathbf{\Lambda}(t)$,
\begin{equation}\label{eq:DecoherenceLambda}
	\mathbf{\Lambda}(t) = \mathrm{diag}\{ e^{-\Gamma_{2,3} (t)} ,\, e^{-\Gamma_{1,3} (t)} ,\, e^{-\Gamma_{1,2} (t)} ,\, 1 \},
\end{equation}
where functions $\Gamma_{j,k}(t) = \Gamma_j (t) + \Gamma_k (t)$ are the antiderivatives,
\begin{equation}\label{eq:GammaDef}
	\Gamma_{j}(t) = \int_{0}^{t} \gamma_{j}(t^\prime) dt^\prime 
\end{equation}
and are all non-negative. Here, again $\Lambda_{44}(t) = 1$ is simply the trace preservation condition. The corresponding Floquet pair $(P_t, e^{tX})$ can then be calculated by finding its matrix counterpart $(\mathbf{P}(t), e^{t\mathbf{X}})$ and transforming back to $B(\matr{2})$. By \eqref{eq:FloquetFormGeneral} and \eqref{eq:XgeneralFormula},
\begin{subequations}
	\begin{equation}\label{eq:XpureDecM4}
		\mathbf{X} = -\frac{1}{T} \, \mathrm{diag}\{\Gamma_{2,3}(T),\,\Gamma_{1,3}(T),\,\Gamma_{1,2}(T),\,0\} ,
	\end{equation}
	\begin{equation}\label{eq:PtMatrixPureDec}
		\mathbf{P}(t) = \mathrm{diag}\{ e^{-\vartheta_{2,3} (t)} ,\, e^{-\vartheta_{1,3} (t)} ,\, e^{\vartheta_{1,2} (t)} ,\, 1 \},
	\end{equation}
\end{subequations}
for functions $\vartheta_{j,k}(t)$ being the shorthand for
\begin{equation}\label{eq:VarthetaDef}
	\vartheta_{j,k} (t) = \Gamma_{j,k}(t) - \frac{\Gamma_{j,k}(T)}{T} t.
\end{equation}
By \eqref{eq:GammaDef}, functions $\Gamma_{j}(t)$ satisfy additivity property $\Gamma_{j}(t+T) = \Gamma_{j}(t) + \Gamma_{j}(T)$ and so $\mathbf{P}(t)$ is periodic. Inverting the vectorization and performing some mild algebra, one recovers original maps over $\matr{2}$,
\begin{subequations}
	\begin{equation}\label{eq:MapsPureDecX}
		X(x) = \left( \begin{array}{cc} -\beta_1 (x_{11}-x_{22}) & \beta_2 x_{21} - \beta_3 x_{12} \\ \beta_2 x_{12} - \beta_3 x_{21} & \beta_1 (x_{11}-x_{22}) \end{array}\right),
	\end{equation}
	\begin{equation}\label{eq:MapsPureDecPt}
		P_t(x) = \left( \begin{array}{cc} \xi_{1}(t)\, x_{11} + \xi_{2}(t) \, x_{22} & \chi_{1}(t)\, x_{12}-\chi_{2}(t)\, x_{21} \\ \chi_{1}(t)\, x_{21} - \chi_{2}(t) \, x_{12} & \xi_{2}(t)\, x_{11} + \xi_{1}(t)\, x_{22} \end{array}\right),
	\end{equation}
\end{subequations}
where the following notation was introduced for brevity,
\begin{align}\label{eq:XiChiDefinitions}
	\beta_{1,2} &= \frac{1}{2T}\left(\Gamma_1 (T) \pm \Gamma_2 (T)\right), \quad \beta_3 = \frac{1}{2T}\left(\beta_1 +2 \Gamma_{3}(T)\right), \\
		\xi_{1,2}(t) &= \frac{1}{2}\left(1\pm e^{-\vartheta_{1}(t)-\vartheta_2 (t)}\right), \\
		\chi_{1,2}(t) &= \frac{1}{2}\left(e^{-\vartheta_{1}(t)-\vartheta_3 (t)}\pm e^{-\vartheta_{2}(t)-\vartheta_3 (t)}\right).
\end{align}
By diagonal structure of \eqref{eq:XpureDecM4}, the eigenbasis of both $X$, $\Lambda_T$ is simply $\varphi_j = F_j$. This gives rise to set of characteristic multipliers
\begin{equation}
	\spec{\Lambda_T} = \{1,\,e^{-\Gamma_{2,3}(T)},\,e^{-\Gamma_{1,3}(T)},\,e^{-\Gamma_{1,2}(T)} \},
\end{equation}
and set $\mathcal{E}_2$ in spectral decomposition \eqref{eq:SpectralDecLambdaT} is empty.
The general solution in this case admits an explicit form \eqref{eq:FloquetGenSol} and can be put as
\begin{equation}\label{eq:PureDecState}
	\rho_t = \frac{1}{\sqrt{2}} \left( P_t (I) + \sum_{\pi \, \text{even}}c_{\pi(1)} e^{-\frac{t}{T}\Gamma_{\pi(2),\pi(3)}(T)} P_t (\sigma_{\pi(1)})\right)
\end{equation}
for even permutations $\pi$ in symmetric group $S_3$. As clearly $\spec{\Lambda_T}\subset\mathbb{D}^1$, all solutions $\rho_j (t), \rho_t$ are stable. Immediately, \eqref{eq:PureDecState} yields a unique periodic limit cycle $\rho_{t}^{\infty}	= \frac{1}{2} I$	being in this case a trivial \emph{limit point} in $\matr{2}$, the \emph{maximally mixed state}. The CP-divisibility of semigroup part $\{e^{tX} : t\in\reals_+\}$ can be shown by checking that the expression \eqref{eq:MapsPureDecX} for map $X$ can be cast into
\begin{equation}\label{eq:RandomDynXmapPauli}
	X(x) = \sum_{j=1}^{3} \frac{\Gamma_j (T)}{T} \left( \sigma_j x \sigma_j - \frac{1}{2}\acomm{\sigma_j \sigma_j}{x} \right)
\end{equation}
which, since $\Gamma_j (t) \geqslant 0$, is of standard form; therefore $\{e^{tX}\}\subset\cptp{\matr{2}}$ and is a CP-divisible contraction semigroup.

Finally, we verify whether equations \eqref{eq:PtCPdivcondition} and \eqref{eq:PtCPsuffCond} of theorem \ref{prop:etXcpDiv} actually correspond to CP-divisibility and complete positivity of $P_t$. This is achieved by finding exact algebraic conditions, which guarantee complete positivity of either $P_t$, or its corresponding propagator $\mathcal{V}_{t,s} = P_t P_{s}^{-1}$, i.e.~by construction and analysis of their Choi matrices. The results, explicitly presented in \ref{sec:propPureDecproof}, show that $\mathcal{V}_{t,s}\in\cptp{\matr{2}}$ for all $t\geqslant s$ in some interval $\mathcal{I}\subset\reals_+$, if and only if
\begin{equation}\label{eq:M2pureDecPtCPdiv}
	\gamma_{j}(t) - \frac{\Gamma_{j}(T)}{T} \geqslant 0
\end{equation}
for all $t\in\mathcal{I}$ and $j \in \{ 1, \, 2,\, 3 \}$, and $P_t \in\cp{\matr{2}}$ for given $t\in\reals_+$ if 
\begin{equation}\label{eq:M2pureDecPtCP}
	\Gamma_{j}(t) - \frac{\Gamma_{j}(T)}{T}t \geqslant 0
\end{equation}
for all $j \in \{ 1, \, 2,\, 3 \}$. These two conditions are then equivalent to claims \ref{claim:PtCPdivcondition} and \ref{claim:PtCPcondition} of theorem \ref{prop:etXcpDiv}.

\subsection{Periodically driven two-level system}
\label{sec:PeriodicTLS}

The second example concerns a \emph{two-level system with periodically modulated Hamiltonian}, coupled to external reservoir via standard ladder operators constructed from Pauli matrices. We utilize the MME in usual standard form \cite{Szczygielski2013}, however with time-dependent Hamiltonian part,
\begin{equation}\label{eq:TLSMME}
	\frac{d\rho_t}{dt} = L_t(\rho_t) = -\frac{i\omega (t)}{2} \comm{\sigma_3}{\rho_t} + \gamma_{\uparrow} D_{\sigma_+} (\rho_t) + \gamma_{\downarrow} D_{\sigma_-}(\rho_t),
\end{equation}
where $D_A$ is defined as $D_A(\rho) = A\rho A^\hadj - \frac{1}{2}\acomm{A^\hadj A}{\rho}$, matrices $\sigma_\pm = \frac{1}{2}(\sigma_1 \pm i\sigma_2)$ are the usual \emph{ladder operators} and $\gamma_{\uparrow}$, $\gamma_{\downarrow} > 0$ stand for pumping and dumping transition rates, respectively. System's self Hamiltonian is $H_t = \frac{1}{2}\omega (t) \sigma_3$ and is diagonal in eigenvectors $e_0 = (0,\,1)$, $e_1 = (1,\,0)$. These eigenvectors denote the \emph{ground} and \emph{exited state}, repectively. Real function $\omega(t)$ is the \emph{energy difference} between states $e_1$ and $e_0$, \emph{periodically modulated} by some external quasi-classical source, $\omega(t) = \omega(t+T)$.

The corresponding Kossakowski matrix of Lindbladian $L_t$ in \eqref{eq:TLSMME} is
\begin{equation}
	\mathbf{a} = \frac{1}{2} \left( \begin{array}{ccc} \gamma_\downarrow+\gamma_\uparrow & i (\gamma_\downarrow-\gamma_\uparrow) & 0 \\ -i (\gamma_\downarrow-\gamma_\uparrow) & \gamma_\downarrow+\gamma_\uparrow & 0 \\ 0 & 0 & 0 \end{array}\right) \geqslant 0,
\end{equation}
which is constant. We next obtain solution in a form of Floquet pair by utilizing the same vectorization procedure as in previous example (we omit calculations for brevity, as the whole procedure is similar),
\begin{subequations}\label{eq:TLSPtX}
	\begin{equation}
		P_t (x) = \left( \begin{array}{cc} x_{11} & e^{-i\varpi(t)} e^{\frac{i\varpi(T)}{T}t} x_{12} \\ e^{i\varpi(t)} e^{-\frac{i\varpi(T)}{T}t} x_{21} & x_{22} \end{array} \right),
	\end{equation}
	\begin{equation}
		X(x) = \left( \begin{array}{cc} -\gamma_\downarrow x_{11} + \gamma_\uparrow x_{2,2} & \left(-\frac{\gamma_\downarrow+\gamma_\uparrow}{2} - i \frac{\varpi(T)}{T}\right)  x_{12} \\ \left(-\frac{\gamma_\downarrow+\gamma_\uparrow}{2} + i \frac{\varpi(T)}{T}\right) x_{21} & \gamma_\downarrow x_{11} - \gamma_\uparrow x_{22} \end{array} \right),
	\end{equation}
\end{subequations}
for antiderivative $\varpi(t) = \int_{0}^{t} \omega(t^\prime) dt^\prime$. With some effort, $X$ can be then put in \emph{standard form}
\begin{equation}
	X = -\frac{i \varpi(T)}{2T} \comm{\sigma_3}{\,\cdot\,\,} + \gamma_\downarrow D_{\sigma_-} + \gamma_\uparrow D_{\sigma_+},
\end{equation}
i.e.~$\{e^{tX} : t\in\reals_+\}$ is CP-divisible. Since $P_t$ does not alter diagonal elements of density matrix and $P_t (x)$ is Hermitian for Hermitian $x$, it is a *-map.

One finds the spectrum of Choi matrix of $P_t$ to be $\{0,2\}$ ($k_0 = 3$), so $P_t\in\cptp{\matr{2}}$ for all $t\in\reals_+$. Curiously, Choi matrix of its propagator, $\mathcal{V}_{t,s} = P_t P_{s}^{-1}$, $t\geqslant s$, yields the same spectrum regardless of $t,s$ so map $P_t$ is CP-divisible \emph{globally}, i.e.~in whole $\reals_+$. This is then confirmed by theorem \ref{prop:etXcpDiv}, since, as $\mathbf{a}$ is \emph{constant}, inequalities \eqref{eq:PtCPdivcondition} and \eqref{eq:PtCPsuffCond} are always satisfied. We remark that this observation remains consistent with theorem \ref{prop:NotEverywhereCPdiv} as global Markovianity of $P_t$ was allowed only if Kossakowski matrix was constant a.e.

Eigendecomposition of matrix counterpart of map $X$ allows also to find $\spec{X}$ and $\spec{\Lambda_T}$, i.e.~sets of \emph{characteristic exponents} and \emph{multipliers},
\begin{subequations}
	\begin{equation}
		\spec{X} = \{ \mu_1 = 0,\, \mu_2 = -\gamma_\downarrow-\gamma_\uparrow,\,\mu_{3,4} = -\frac{1}{2}(\gamma_\downarrow+\gamma_\uparrow)\pm i\frac{\varpi(T)}{T}\},
	\end{equation}
	\begin{equation}
		\spec{\Lambda_T} = \{ \lambda_1 = 1,\, \lambda_2 = e^{-T(\gamma_\downarrow+\gamma_\uparrow )},\, \lambda_{3,4} = e^{-\frac{T}{2}(\gamma_\downarrow+\gamma_\uparrow)}e^{\pm i\varpi(T)} \},
	\end{equation}
\end{subequations}
along with eigenvectors (put in corresponding order)
\begin{equation}
	\varphi_1 = \frac{\sqrt{2}}{\gamma_\downarrow+\gamma_\uparrow} \, \mathrm{diag}\{\gamma_\uparrow ,\, \gamma_\downarrow\}, \quad \varphi_2 = \frac{1}{\sqrt{2}}\sigma_3, \quad \varphi_{3,4} = \pm i\sqrt{2} \, \sigma_{\mp}. 
\end{equation}
Hence, subset $\mathcal{E}_2$ of $\spec{X}$ is again empty.
We emphasize here, that the eigenbasis $\{\varphi_j\}$ is \emph{not} orthogonal (w.r.t. Frobenius inner product) since $\mathbf{X}$ is not normal. Again, $\spec{\Lambda_T}\subset\mathbb{D}^1$ and is closed under complex conjugation. All eigenvectors apart from $\varphi_1$, i.e.~those spanning eigenspaces $E_{\Lambda_T} (\lambda)$ for $\lambda\neq 1$, are then traceless and non-positive semi-definite, as proposition \ref{prop:VarphiProperties} states. Two real multipliers $\lambda_{1,2}$ are simple eigenvalues and so $\varphi_{1,2}$ are Hermitian; naturally, $\varphi_1 \geqslant 0$ and $\varphi_3 = \varphi_{4}^{\hadj}$, as $\lambda_3 = \overline{\lambda_4}$.

An actual solution is then obtained with formulas \eqref{eq:FloquetGenSol} and \eqref{eq:RhoPeriodicSteady},
\begin{align}
	\rho_t &= c_1 \varphi_1 + c_2 e^{-(\gamma_\downarrow+\gamma_\uparrow)t} \varphi_2 \\
	&+ e^{-\frac{t}{2}(\gamma_\downarrow+\gamma_\uparrow)}\left(c_3 e^{\frac{i\varpi(T)}{T}t} \phi_3 (t) + c_4 e^{\frac{-i\varpi(T)}{T}t} \phi_4 (t) \right),\nonumber
\end{align}
where Floquet states $\phi_{3,4}(t) = P_t (\varphi_{3,4})$ are explicitly defined as
\begin{equation}\label{eq:TLSsolution}
	\phi_3 (t) = -i\sqrt{2} \, e^{-i\varpi(t)}e^{\frac{i\varpi(T)}{T}t} \sigma_+, \quad \phi_4 (t) = \phi_3 (t)^\hadj = i\sqrt{2} \, e^{i\varpi(t)}e^{-\frac{i\varpi(T)}{T}t} \sigma_- 
\end{equation}
and coefficients $c_j$ are found to be
\begin{equation}
	c_1 = \frac{1}{\sqrt{2}}, \quad c_2 = \frac{\gamma_\uparrow \sqrt{2}}{\gamma_\downarrow+\gamma_\uparrow}-\sqrt{2} \, \rho_{11}(0), \quad c_3 = \frac{i}{\sqrt{2}}\rho_{12}(0) = \overline{c_4},
\end{equation}
where trace normalization and Hermiticity of $\rho_0$ were implicitly used. Clearly, solution \eqref{eq:TLSsolution} is stable and the asymptotic periodic orbit $\rho_{t}^{\infty}$ in this case is, similarly to previous example, also a single \emph{limit point}, $\rho_{t}^{\infty} = \frac{1}{\sqrt{2}} \varphi_1$.

\section{Note on the non-commutative case}
\label{sec:Counterexample}

Theorems \ref{prop:etXcpDiv} and \ref{prop:NotEverywhereCPdiv} allow to characterize CP-divisibility properties of Floquet normal form in commutative case. It is then natural to ask whether these results possibly could be extended onto general class of \emph{non-commutative} Lindbladians, i.e.~time-dependent maps $L_t$ not subject to condition \eqref{eq:CommCondition}. This is answered negatively in this section by brief examination of simple, numerical counterexample in algebra $\matr{3}$. Namely, we consider a $2\pi$-periodic Lindbladian of general standard form \eqref{eq:Lindbladian2} for $\{F_i\}$ being a Frobenius orthonormal basis of $\matr{3}$ (see \ref{sec:M3basis} for details) and of the Kossakowski matrix $\mathbf{a}_t = [a_{ij}(t)]$, $i,j \in \{1,\, ... ,\, 9\}$ given by equalities
\begin{align}
	a_{22}(t) = a_{77}(t) = 1, \quad a_{55}(t) = 2, \quad a_{88}(t) = 1+\cos{t}, \\
	a_{52}(t) = -i \cos{t}, \quad a_{72}(t) = -i, \quad a_{75}(t) = \cos{t}, \nonumber \\
	a_{25}(t) = \overline{a_{52}(t)}, \quad a_{27}(t) = \overline{a_{72}(t)}, \quad a_{57}(t) = a_{75}(t), \nonumber
\end{align}
with all remaining $a_{ij}(t) = 0$; also, we put $H_t = 0$ for simplicity. By direct check, $\mathbf{a}_t$ is then positive semi-definite. The solution of Master Equation is found by applying again the vectorization scheme and solving a resulting matrix ODE of a form $\frac{d}{dt}\mathbf{\Lambda}_t = \mathbf{L}_t\mathbf{\Lambda}_t$, $\mathbf{\Lambda}_0 = I$ numerically for time-dependent matrix $\mathbf{\Lambda}_t$ (for clarity, we do not include explicit numerical results in the paper). Map $X$, and the semigroup part $e^{tX}$ in consequence, are then found as in \eqref{eq:XgeneralFormula} by calculating a proper matrix logarithm of monodromy matrix $\mathbf{\Lambda}_{2\pi}$ and reverting the vectorization; likewise, the periodic part $P_t$ of Floquet pair is revealed by computing $\Lambda_t e^{-tX}$. An interesting result then is observed: while both maps $P_t$, $e^{tX}$ are trace and hermiticity preserving, none of them is actually completely positive anywhere in $\reals_+ \setminus 2\pi\integers$, nor CP-divisible (their composition $\Lambda_t$ remains globally completely positive and CP-divisible as a quantum dynamics). We show this explicitly by plotting the time evolution of their spectra in figure \ref{fig:Spectra}.
\begin{figure}[htbp]
	\centering
		\includegraphics[width=1.00\textwidth]{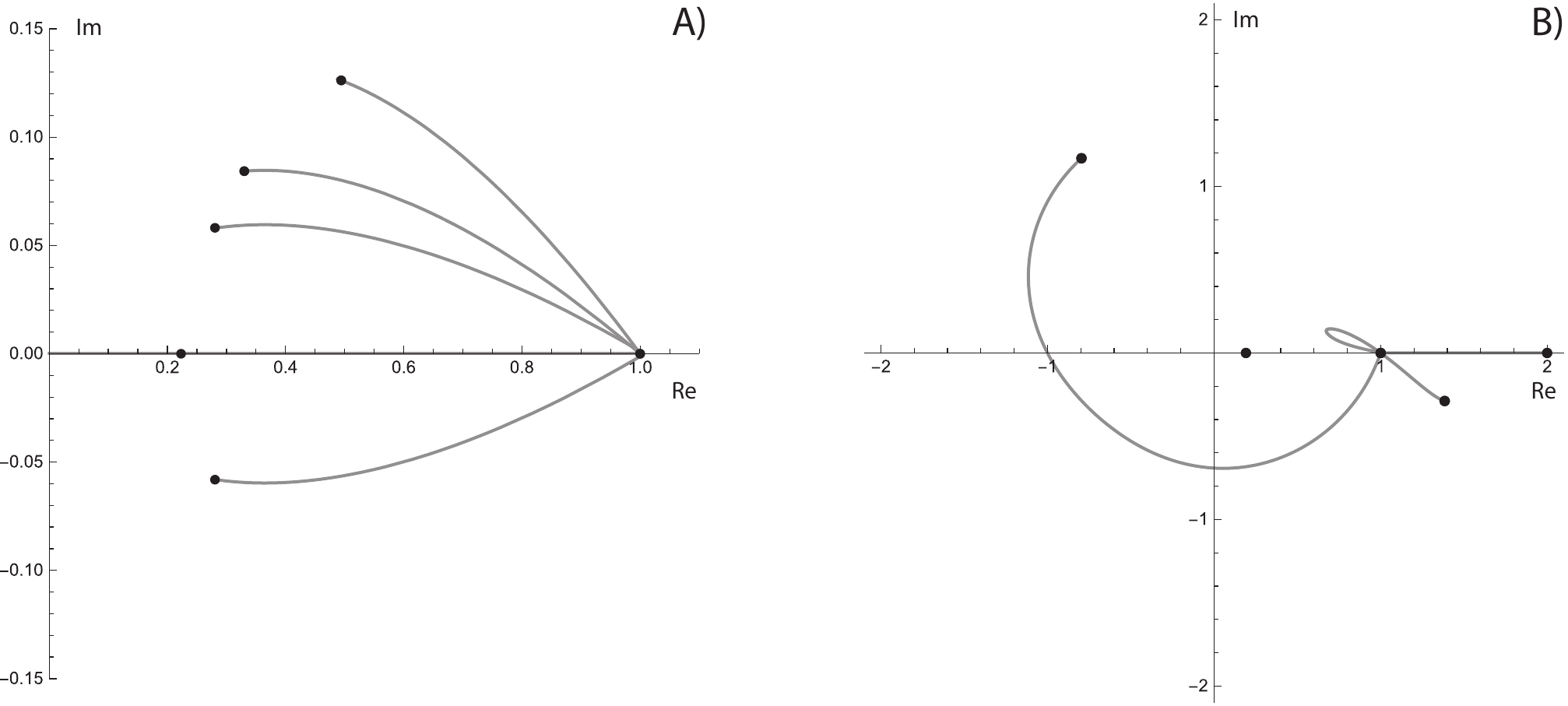}
	\caption{Complex plane plots of eigenvalues (dots) and their trajectories (curves) of maps $e^{tX}$ (image A, spectrum at time $t=0.5$, trajectories drawn for $0\leqslant t \leqslant 0.5$) and $P_t$ (image B, spectrum at $t=4.32$, trajectories for $0\leqslant t \leqslant 4.32$).}
	\label{fig:Spectra}
\end{figure}
Since clearly both spectra are not invariant w.r.t.~complex conjugation, none of the two maps of Floquet pair are completely positive. This fact is further confirmed by checking semi-definiteness of corresponding Choi matrices. Therefore, it is evident that Theorems \ref{prop:etXcpDiv}, \ref{prop:NotEverywhereCPdiv} do not admit a direct application in non-commutative setting as even the semigroup part of the solution may fail to be completely positive. The same can be then stated on CP-divisibility of $P_t$, since its propagator $P_t P_{s}^{-1}$, $s\in [0,t]$, is not completely positive either.

\section{Conclusions}

We presented an insight into general applicability of Floquet theory in description of Markovian Master Equations given by periodic, finite-dimensional Lindbladians in standard form. The performed analysis allowed for formulating some remarks on Floquet normal form of the induced quantum dynamical maps, partially in general case, and especially in simplified case of commutative Lindbladian families. In particular, it was shown that in generic case of periodic $L_t$, it is impossible for both maps of the Floquet pair to be globally simultaneously Markovian in commutative case. It was also shown that the traditional results of Floquet theory, like analysis of stability based on characteristic multipliers of the system, still possesses an excellent application in case of completely positive dynamics. Two examples of possible non-trivial physical applicability of such Floquet-Lindblad theory were also briefly examined. However, the general case of \emph{non-commutative} Lindbladian families remains an open problem requiring more involved study, since, interestingly, no \emph{global} Markovianity, nor even complete positivity of the Floquet normal form is guaranteed once the commutativity condition is abandoned.

\section*{Acknowledgments}
Author expresses his thanks to prof.~Robert Alicki for valuable comments and stimulating discussions, as well as to anonymous Reviewer for constructive suggestions, which led to substantial improvement of a preliminary version of the manuscript. Support by the National Science Centre, Poland, via grant No. 2016/23/D/ST1/02043 is greatly acknowledged.

\appendix

\section{Mathematical supplement}

\label{sec:SecondaryResults}

\begin{lemma}\label{lemma:EveryLinMap}The following hold for every linear *-map $T$ on $\matrd$:
	\begin{enumerate*}[label={\alph*)}]
		\item\label{point:EveryLinMapA} $T$ admits a unique Hermitian matrix $[t_{jk}]\in\matr{d^2}$, such that $T(x) = \sum_{j,k=1}^{d^2} t_{jk} F_j x F_{k}^{\hadj}$ for every $x\in\matrd$; 
		\item\label{point:EveryLinMapB} $T$ is completely positive iff $[t_{jk}]\geqslant 0$;
		\item\label{point:EveryLinMapC} if $T$ is trace preserving, then there exist Hermitian matrices $G,K\in\matrd$ such that\end{enumerate*}
		\begin{equation}\label{eq:HpTpMapForm}
			T(x) = x + i \comm{G}{x} -\acomm{K}{x} + \sum_{j,k=1}^{d^2-1} t_{jk} F_j x F_{k}^{\hadj}.
		\end{equation}
\end{lemma}

\begin{proof}
Structure theorems by de Pillis \cite{Pillis1967}, Jamio\l kowski \cite{Jamiolkowski1972}, Choi \cite{Choi_1975} and Hill \cite{Hill1973,Poluikis1981} allow to represent any *-map $T$ in a form $T(x) = \sum_{i} \alpha_i X_i x X_{i}^{\hadj}$, where $X_i \in \matrd$ and $\alpha_i \in \reals$ ($\alpha_i \geqslant 0$ iff $T$ is completely positive). It suffices to expand $X_i = \sum_j x_{i,j} F_j$ in Frobenius basis and collect expansion coefficients in form of new matrix, $t_{jk} = \sum_{i} x_{i,j} \overline{x_{i,k}}$. Claims \ref{point:EveryLinMapA} and \ref{point:EveryLinMapB} then follow by examining properties of $[t_{jk}]$. For \ref{point:EveryLinMapC}, splitting sums in general decomposition of $T$ allows one to write
\begin{equation}
	T(x) = Ex + xE^\hadj + \Psi(x),
\end{equation}
for $\Psi (x) = \sum_{j,k=1}^{d^2-1} t_{jk} F_j x F_{k}^{\hadj}$ and $E = \frac{1}{2d} t_{d^2 d^2} \cdot I + \frac{1}{\sqrt{d}} \sum_{j=1}^{d^2-1} t_{jd^2} F_j$, where we employed hermiticity of $[t_{jk}]$. $E$ admits a unique Cartesian decomposition $E = M + iN$, where $M = \frac{1}{2} (E+E^\hadj)$ and $N = \frac{1}{2i} (E-E^\hadj)$ are both Hermitian; therefore
\begin{equation}\label{eq:T2}
	T(x) = i\comm{N}{x} + Mx + xM + \Psi(x).
\end{equation}
Trace preservation condition imposed on \eqref{eq:T2} and cyclicity of trace imply $M = 2^{-1} I + M^\prime$ for $M^\prime = 2^{-1}\sum_{j,k=1}^{d^2-1} t_{jk} F_{k}^{\hadj} F_j$; after substituting back to \eqref{eq:T2} and identifying $G = N$ and $K = M^\prime$, it yields formula \eqref{eq:HpTpMapForm}.
\end{proof}

\subsection{Properties of map \texorpdfstring{$P_t$}{P_t} in random dynamics example}
\label{sec:propPureDecproof}

Here we provide justification for conditions \eqref{eq:M2pureDecPtCPdiv} and \eqref{eq:M2pureDecPtCP}, which are sufficient and necessary for complete positivity and Markovianity of map $P_t$ \eqref{eq:MapsPureDecPt}. Proof will rely on determining geometrical conditions for positivity of certain Choi matrices, however with crucial help from infinite divisibility assumption of Markovian dynamics. For the following result, let us define a vector-valued function $\vec{\vartheta} : \reals_+ \to \reals^3$,
\begin{equation}
	\vec{\vartheta}(t) = (\vartheta_1 (t),\vartheta_2 (t),\vartheta_3 (t)), \quad \vartheta_j (t) = \gamma_j (t) - \frac{\Gamma_{j}(T)}{T}.
\end{equation}

\begin{proposition}\label{prop:PureDec}
Map $P_t$ \eqref{eq:MapsPureDecPt} yielded by equation \eqref{eq:pureDecohMME} over $\matr{2}$ satisfies:
\begin{enumerate}
	\item\label{point:PureDecFour} $P_t \in \cptp{\matr{2}}$ iff $\vec{\vartheta} (t) \in \mathcal{A} = \bigcup_{j=1}^{3}\mathcal{A}_j$, where $\mathcal{A}_j \subset \reals^3$ are unbounded regions,
	\begin{subequations}\label{eq:PureDecThree}
		\begin{equation}\label{eq:regionA1}
			\mathcal{A}_1 = \{(x,y,z) : x,y \geqslant 0, \, z \geqslant \ln{\frac{\cosh{\frac{1}{2}(x-y)}}{\cosh{\frac{1}{2}(x+y)}}}\},
		\end{equation}
		\begin{equation}
			\mathcal{A}_2 = \{(x,y,z) : x > x+y > 0, \, z \geqslant \ln{\frac{\sinh{\frac{1}{2}(x-y)}}{\sinh{\frac{1}{2}(x+y)}}}\},
		\end{equation}
		\begin{equation}
			\mathcal{A}_3 = \{(x,y,z) : y > x+y > 0, \, z \geqslant \ln{\frac{\sinh{\frac{1}{2}(y-x)}}{\sinh{\frac{1}{2}(x+y)}}}\};
		\end{equation}
	\end{subequations}
	\item\label{point:PureDecThree} $P_t$ is CP-divisible in some interval $\mathcal{I}\subset [0,T)$ iff $\vec{\vartheta}(t)-\vec{\vartheta}(s) \in \reals_{+}^{3}$ for all $t,s \in \mathcal{I}$, $t \geqslant s$, which is the case iff $\gamma_{j}(t) - \frac{\Gamma_{j}(T)}{T} \geqslant 0$ for $j \in \{ 1, \, 2,\, 3 \}$ and for all $t\in \mathcal{I}$.
\end{enumerate}
\end{proposition}

\begin{proof}
For claim \ref{point:PureDecFour}, calculate the Choi matrix of $P_t$,
\begin{equation}\label{eq:PtDecohChoi}
	\mathcal{C}[P_t] = \left( \begin{array}{cccc} \xi_1(t) & 0 & 0 & \chi_1 (t) \\ 0 & \xi_2 (t) & -\chi_2 (t) & 0 \\ 0 & -\chi_2(t) & \xi_2(t) & 0 \\ \chi_1 (t) & 0 & 0 & \xi_1 (t) \end{array}\right),
\end{equation}
as well as its spectrum,
\begin{equation}
	\spec{\mathcal{C}[P_t]} = \{ \xi_1 (t)\pm\chi_1(t), \xi_2 (t)\pm\chi_2(t) \},
\end{equation}
where $\xi_{1,2}(t)$ and $\chi_{1,2}(t)$ were defined by \eqref{eq:XiChiDefinitions}. Then, $\mathcal{C}[P_t]\geqslant 0$ iff $\spec{\mathcal{C}[P_t]}\subset\reals_+$. Introducing variables $\alpha_j = e^{-\vartheta_j (t)}$ for $j\in\{1,\,2,\,3\}$, one can check by hand that non-negativity of $\spec{\mathcal{C}[P_t]}$ yields a system of four linear inequalities
\begin{equation}\label{eq:PtCPdecohIneq}
	\left\{
	\begin{aligned} 
		&\alpha_1\alpha_2+\alpha_1 \alpha_3-\alpha_2 \alpha_3 \leqslant 1,\\
		&\alpha_1\alpha_2-\alpha_1 \alpha_3+\alpha_2 \alpha_3 \leqslant 1, \\
		&-\alpha_1\alpha_2+\alpha_1 \alpha_3+\alpha_2 \alpha_3 \leqslant 1, \\
		&-\alpha_1\alpha_2-\alpha_1 \alpha_3-\alpha_2 \alpha_3 \leqslant 1.
	\end{aligned}
	\right.
\end{equation}
Solution of this system may be then divided into three unbounded regions,
\begin{subequations}
	\begin{equation}
		\mathcal{B}_1 = \{ \alpha_1,\alpha_2 \leqslant 1, \, \alpha_3\leqslant \tfrac{1+\alpha_1\alpha_2}{\alpha_1 + \alpha_2} \},
	\end{equation}
		\begin{equation}
		\mathcal{B}_2 = \{ \alpha_1 < 1, \, 1 < \alpha_2 < \alpha_{1}^{-1}, \, \alpha_3\leqslant \tfrac{-1+\alpha_1\alpha_2}{\alpha_1-\alpha_2} \},
	\end{equation}
	\begin{equation}
		\mathcal{B}_3 = \{ \alpha_1 > 1, \, \alpha_2 < \alpha_{1}^{-1}, \, \alpha_3\leqslant \tfrac{1-\alpha_1\alpha_2}{\alpha_1-\alpha_2} \}.
	\end{equation}
\end{subequations}
Put also $\mathcal{B} = \mathcal{B}_1\cup\mathcal{B}_2\cup\mathcal{B}_3$. By reverting the $\alpha_j$ substitution, regions $\mathcal{A}_j$ show up as preimages of $\mathcal{B}_j$ under a mapping $(x,y,z) \mapsto (e^{-x},e^{-y},e^{-z})$. A schematic plot of region $\mathcal{A}$ is also presented in figure \ref{fig:Aset}.

\begin{figure}[htbp]
	\centering
		\includegraphics[width=0.5\textwidth]{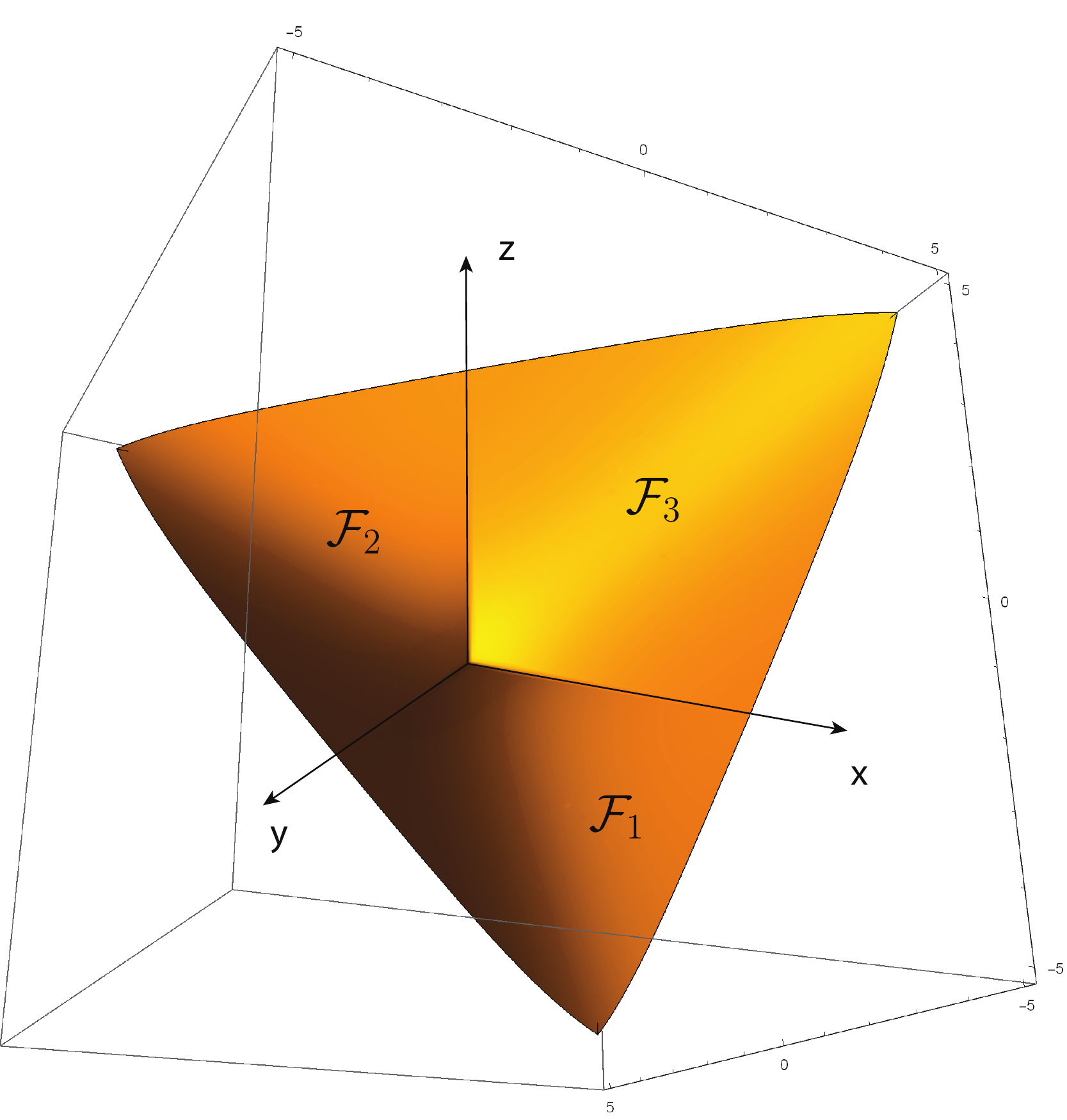}
	\caption{Schematic plot of region $\mathcal{A}$ for $x,y,z \in [-5,5]$. All bounding surfaces $\mathcal{F}_{1,2,3}$ are tangent to appropriate planes $x,y,z=0$ at the origin. The whole region is invariant w.r.t. rotations by angle $2\pi n / 3$, $n\in\integers$, around axis $x=y=z$.}
	\label{fig:Aset}
\end{figure}

Finally, claim \ref{point:PureDecThree} involves checking whether the \emph{propagator} of $P_t$, defined by simple expression $\mathcal{V}_{t,s} = P_t P_{s}^{-1}$, is completely positive. This is achieved by computing its matrix counterpart $\mathbf{P}(t)\mathbf{P}(s)^{-1}$ and transforming to $B(\matr{2})$. The result can be shown to be, due to diagonal structure of $\mathbf{P}(t)$, similar to \eqref{eq:MapsPureDecPt},
\begin{equation}
	\mathcal{V}_{t,s}(x) = \left( \begin{array}{cc} \xi_{1}(t,s)\, x_{11} + \xi_{2}(t,s) \, x_{22} & \chi_{1}(t,s)\, x_{12}-\chi_{2}(t,s)\, x_{21} \\ \chi_{1}(t,s)\, x_{21} - \chi_{2}(t,s) \, x_{12} & \xi_{2}(t,s)\, x_{11} + \xi_{1}(t,s)\, x_{22} \end{array}\right),
\end{equation}
with a new set of two-variable functions $\xi_{1,2}$, $\chi_{1,2}$ and $\delta_{j}$,
\begin{subequations}
	\begin{equation}
		\xi_{1,2}(t,s) = \frac{1}{2}\left(1\pm e^{-\delta_1 (t,s)-\delta_2 (t,s))}\right),
	\end{equation}
	\begin{equation}
		\chi_{1,2}(t,s) = \frac{1}{2}\left(e^{-\delta_{1}(t,s)-\delta_3 (t,s))}\pm e^{-\delta_{2}(t,s)-\delta_3 (t,s))}\right),
	\end{equation}
	\begin{equation}
		\delta_j (t,s) = \vartheta_j(t) - \vartheta_j (s), \quad j \in \{1,\,2,\,3\}.
	\end{equation}
\end{subequations}
By its similarity to \eqref{eq:MapsPureDecPt}, Choi matrix of $\mathcal{V}_{t,s}$ is of almost the same form as \eqref{eq:PtDecohChoi}, however with $\delta_j (t,s)$ in place of $\vartheta_j (t)$. Requiring non-negativity of its spectrum leads, by introducing variables $\alpha_j = e^{-\delta_j (t,s)}$, to exactly the same system of inequalities as \eqref{eq:PtCPdecohIneq}. Therefore, $\mathcal{C}[\mathcal{V}_{t,s}] \geqslant 0$ and $\mathcal{V}_{t,s}\in\cptp{\matr{2}}$ if and only if $(\alpha_j)\in\mathcal{B}$, or equivalently, if
\begin{equation}\label{eq:DeltaCond1}
	\vec{\delta}(t,s) = \vec{\vartheta}(t)-\vec{\vartheta}(s) = (\delta_1 (t,s), \delta_2 (t,s), \delta_3 (t,s)) \in \mathcal{A}.
\end{equation}
\noindent However, the requirement of \emph{divisibility} allows to greatly refine condition \eqref{eq:DeltaCond1}. First, notice that each $P_t \in \cptp{\matr{2}}$ is uniquely described by a \emph{vector} $\vec{\vartheta}(t)\in\mathcal{A}$ and function $t\mapsto P_t$ is represented by differentiable \emph{curve} $t \mapsto \vec{\vartheta}(t)$. Similarly, every map of a form $P_t P_{s}^{-1}$ is bijectively determined by a vector $\vec{\delta}(t,s)$, with $\vec{\delta}(t,t)$ corresponding to identity map for any $t\geqslant 0$. Geometrically, function $t\mapsto \vec{\delta}(t,s)$ for some constant $s$ is also a curve, created by \emph{translating} curve $\vec{\vartheta}$ by constant vector $-\vec{\vartheta}(s)$, such that point $\vec{\vartheta}(s)$ is mapped into point $\vec{0} = (0,0,0)$, the origin. With any such curve, one associates its \emph{velocity},
\begin{equation}
	\vec{v}(t) = \frac{d\vec{\vartheta}(t)}{dt} = \frac{d\vec{\delta}(t,s)}{dt},
\end{equation}
which is \emph{tangent} to it at point $\vec{\vartheta}(t)$. Suppose now $P_t$ is CP-divisible in some interval $\mathcal{I}\subset\reals_{+}$. Then, for arbitrarily chosen $t,t^\prime,s\in\mathcal{I}$ such that $t^\prime \in [s,t]$, propagator $\mathcal{V}_{t,s}$ is a composition of two subsequent propagators, $\mathcal{V}_{t,s} = \mathcal{V}_{t, t^\prime}\mathcal{V}_{t^\prime,s}$, both of them being again completely positive and divisible. As such, they are both uniquely described by some vectors $\vec{\delta}(t,t^\prime),\vec{\delta}(t^\prime, s) \in \mathcal{A}$. Divisibility condition is then equivalent to the addition rule
\begin{equation}
	\vec{\delta}(t,s) = \vec{\delta}(t,t^\prime)+\vec{\delta}(t^\prime,s), \quad t^\prime \in [s,t].
\end{equation}
Suppose that the curve $\vec{\vartheta}$ corresponding to $P_t$ is s.t. any component of its velocity, $v_j (t)$, is negative anywhere in $\reals_+$. Then, as $t\mapsto\vec{\vartheta}(t)$ is continuous, there exists an interval $[t_1, t_2]$ such that $v_j (t_1) = v_j (t_2) = 0$ and $v_j (t) < 0$ for all $t \in (t_1, t_2)$, i.e.~$\vec{v}(t)$ points in the direction \emph{outside} of set $\mathcal{A}$ within $(t_1,t_2)$. Take any fixed $s \in (t_1, t_2)$; necessarily, $v_j (s) < 0$. Then, a curve $\vec{\delta}(\cdot,s)$, starting at $\vec{0}$ is a geometrical representation of $\mathcal{V}_{t,s}$ for $t \geqslant s$, as mentioned earlier. However, the velocity vector at the origin $\vec{v}(s)\notin\reals_{+}^{3}$ and so the curve $\vec{\delta}(\cdot,s)$ is initially directed outside of $\reals_{+}^{3}$, i.e.~there surely exists some $t^\prime > s$ small enough such that $\vec{\delta}(t^\prime,s) \notin \reals_{+}^{3}$. Moreover, it can be also shown that even $\vec{\delta}(t^\prime, s)\notin \mathcal{A}$; to achieve this, consider one of the boundary surfaces of region $\mathcal{A}$ along one of the axes. Since $\mathcal{A}$ is invariant with respect to rotations by angle $2n\pi /3$, $n\in\integers$ around axis $x=y=z$, without loss of generality we can take the surface $\mathcal{F}_1$, the lowest boundary of sub-region $\mathcal{A}_1$ \eqref{eq:regionA1}. Definition of $\mathcal{A}$ yields that $\mathcal{F}_1$ can be represented as a \emph{function} $\mathcal{F}_1 (x,y)$ given by formula
\begin{equation}
	\mathcal{F}_1 (x,y) = \ln{\frac{\cosh{\frac{1}{2}(x-y)}}{\cosh{\frac{1}{2}(x+y)}}}, \quad \mathcal{F}_1 : \reals_{+}^{2} \to [0,-\infty).
\end{equation}
It is easy to notice $\lim_{(x,y)\to (0,0)}\mathcal{F}_1 (x,y) = 0$, with both $x$, $y$ tending to zero from above. Let us consider any plane $\mathcal{P}_{\vec{n}}$ containing the $z$ axis, spanned by vector $(0,0,1)$ and any vector $\vec{n} = (n_x,n_y,0)$, $n_x, n_y \geqslant 0$, laying in plane $z=0$ (see fig. \ref{fig:Curve}). 

\begin{figure}[htbp]
	\centering
		\includegraphics[width=0.6\textwidth]{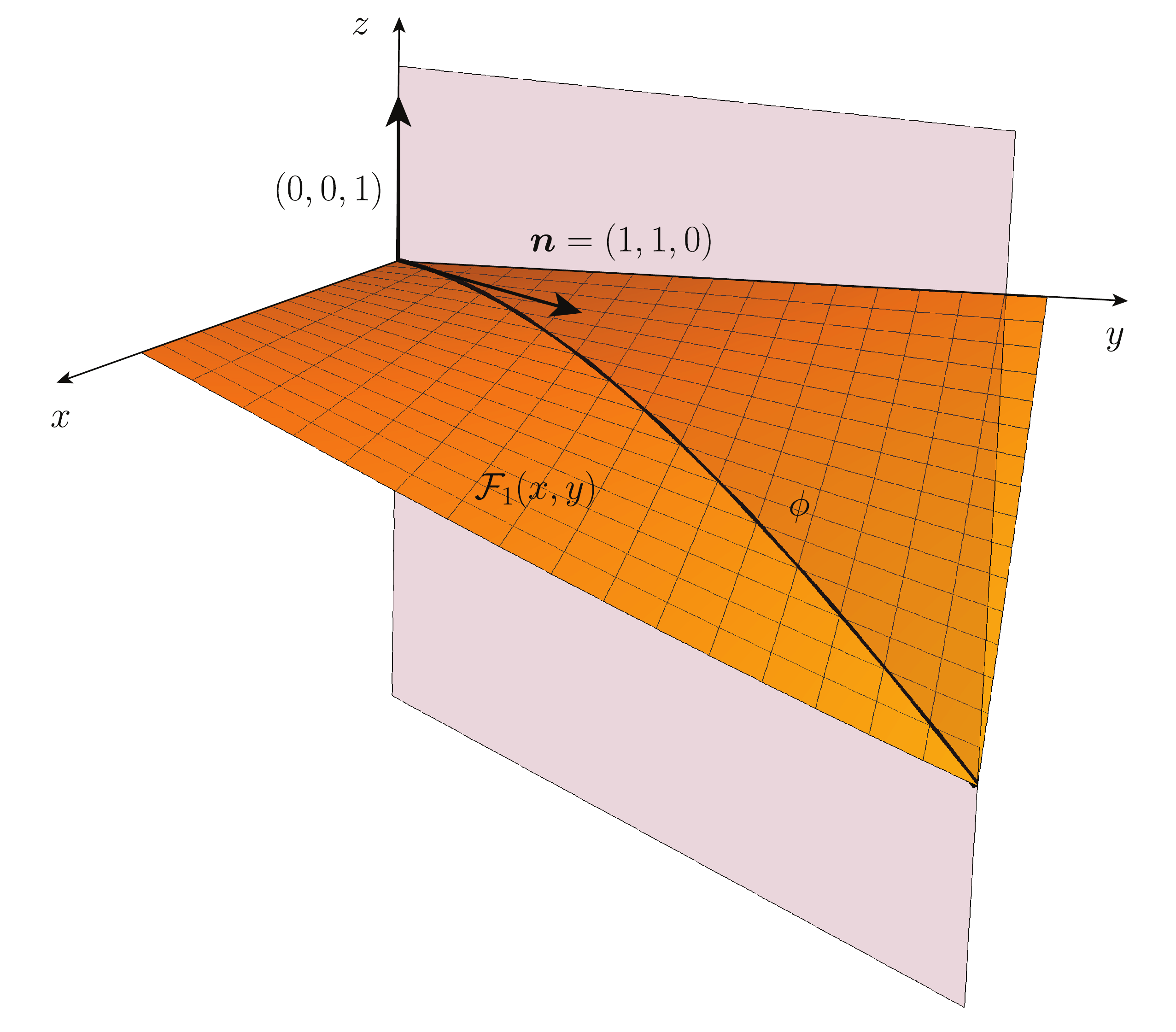}
	\caption{Plot of boundary surface $\mathcal{F}_1 (x,y)$ in close proximity of the origin and exemplary plane $\mathcal{P}_{\vec{n}}$ for $\vec{n} = (1,1,0)$. Intersection of $\mathcal{F}_1 (x,y)$ and $\mathcal{P}_{\vec{n}}$ defines a convex curve $\vec{\phi}$, the velocity of which is simply $\vec{n}$ at the origin, i.e.~is tangent both to the surface and to plane $z=0$ (regardless of chosen $\mathcal{P}_{\vec{n}}$).}
	\label{fig:Curve}
\end{figure}

Intersection of $\mathcal{P}_{\vec{n}}$ and $\mathcal{F}_1$ defines a \emph{convex curve} $\vec{\phi} = (\phi_1, \phi_2, \phi_3)$ which may be given in parametric form as
\begin{equation}
	\phi_1 (\xi) = n_x \xi, \quad \phi_2 (\xi) = n_y \xi, \quad \phi_3 (\xi) = \mathcal{F}_1 (x(\xi), y(\xi)) = \ln{\frac{\cosh{\frac{\xi}{2}(n_x-n_y)}}{\cosh{\frac{\xi}{2}(n_x+n_y)}}}.
\end{equation}
One can check, that the velocity vector $\frac{d\vec{\phi}(\xi)}{d\xi}$ of curve $\vec{\phi}$ simply evaluates to $\vec{n}$ for $\xi=0$. In consequence, all vectors tangent to $\mathcal{F}_1$ at $\vec{0}$ are also tangent to the plane $z=0$. Likewise, all vectors tangent to surfaces $\mathcal{F}_2$ and $\mathcal{F}_3$ at $\vec{0}$ are also tangent to planes $x=0$ and $y=0$, respectively. Therefore, if velocity $\vec{v}(s)$ of curve $\vec{\delta}(\cdot,s)$ at $\vec{0}$ has negative $j$-th component, then point $t^\prime > s$ can be chosen in such way that segment of curve $\vec{\delta}(t,s)$ for $t \in [s,t^\prime]$ is not enclosed by surface $\mathcal{F}_j$, and in the result, not in $\mathcal{A}$. In consequence, $\vec{\delta}(t^\prime,s)\notin\mathcal{A}$ and so $\mathcal{V}_{t^\prime,s}\notin\cptp{\matr{2}}$. We have therefore found a division $\mathcal{V}_{t,s} = \mathcal{V}_{t,t^\prime}\mathcal{V}_{t^\prime,s}$ such that at least one of the propagators at the r.h.s. fails to be completely positive; therefore, $P_t$ cannot be CP-divisible. From this we imply that a curve $\vec{\vartheta}$ can represent a CP-divisible map iff $\vec{v}(t)\in\reals_{+}^{3}$ for all $t\in\mathcal{I}$, i.e.~if condition
\begin{equation}
	\frac{d\vartheta_j (t)}{dt} = \gamma_{j}(t) - \frac{\Gamma_{j}(T)}{T} \geqslant 0,
\end{equation}
holds for all $t\in\mathcal{I}$ and $j \in \{1,\, 2,\, 3\}$. This concludes the proof.
\end{proof}

\subsection{Frobenius orthonormal basis of algebra \texorpdfstring{$\matr{3}$}{M3}}
\label{sec:M3basis}

The following matrices $F_i$, $i\in\{1,\,... \, , \, 9\}$, were used as a basis of $\matr{3}$ while conducting numerical analysis outlined in section \ref{sec:Counterexample}. It is straightforward to check that $\tr{F_{i}^{\hadj} F_j} = \delta_{ij}$, i.e.~the basis is Frobenius orthonormal; one often finds such matrices in literature as generators of $\mathrm{SU}(N)$ or so-called \emph{Gell-Mann matrices} (up to normalizing factors; see \cite{IngemarBengtsson2017}).

\small
\begin{align}
	&F_1 = \frac{1}{\sqrt{2}}\left(\begin{array}{ccc} 0 & 1 & 0 \\ 1 & 0 & 0 \\ 0 & 0 & 0 \end{array}\right), \quad F_2 = \frac{1}{\sqrt{2}}\left(\begin{array}{ccc} 0 & 0 & 1 \\ 0 & 0 & 0 \\ 1 & 0 & 0 \end{array}\right), \quad F_3 = \frac{1}{\sqrt{2}}\left(\begin{array}{ccc} 0 & 0 & 0 \\ 0 & 0 & 1 \\ 0 & 1 & 0 \end{array}\right), \nonumber \\
	&F_4 = \frac{i}{\sqrt{2}}\left(\begin{array}{ccc} 0 & -1 & 0 \\ 1 & 0 & 0 \\ 0 & 0 & 0 \end{array}\right), \quad F_5 = \frac{i}{\sqrt{2}}\left(\begin{array}{ccc} 0 & 0 & -1 \\ 0 & 0 & 0 \\ 1 & 0 & 0 \end{array}\right), \quad F_6 = \frac{i}{\sqrt{2}}\left(\begin{array}{ccc} 0 & 0 & 0 \\ 0 & 0 & -1 \\ 0 & 1 & 0 \end{array}\right), \nonumber \\
	&F_7 = \frac{1}{\sqrt{2}}\left(\begin{array}{ccc} 1 & 0 & 0 \\ 0 & -1 & 0 \\ 0 & 0 & 0 \end{array}\right), \quad F_8 = \frac{1}{\sqrt{6}}\left(\begin{array}{ccc} 1 & 0 & 0 \\ 0 & 1 & 0 \\ 0 & 0 & -2 \end{array}\right), \quad F_9 = \frac{1}{\sqrt{3}}\left(\begin{array}{ccc} 1 & 0 & 0 \\ 0 & 1 & 0 \\ 0 & 0 & 1 \end{array}\right). \nonumber
\end{align}
\normalsize


\end{document}